%% file: complcut2.bezk.bezn.tex
\documentclass{fac}
\pdfoutput=1

\usepackage{ifthen}
\newboolean{commentsaon}         %
\setboolean{commentsaon}{true}
 \setboolean{commentsaon}{false}  %

\newboolean{commentson} %
\setboolean{commentson}{true}
 \setboolean{commentson}{false} %

\newcommand{\comment}[1]
{\ifthenelse{\boolean{commentson}\AND\boolean{commentsaon}}
   {{\par\noindent\mbox{}{\small\color{blue}[ *** #1 ]\par}\noindent\par}}{}}

\newcommand{\commenta}[1]
{\ifthenelse{\boolean{commentsaon}}
   {{\par\noindent\mbox{}{\small\color[rgb]{0, .5, 0}[ *** #1 ]\par}\noindent\par}}{}}
\renewcommand*{\today}{16 February 2016}

\usepackage{xspace}
\usepackage{latexsym}
\usepackage{amssymb}   %
\usepackage{color}
\usepackage{bm}

\usepackage{pgf}
\usepackage{hyperref}
\hypersetup{bookmarksdepth=2, bookmarksopenlevel=2}

\usepackage[hyphenbreaks]{breakurl}

\newtheorem{theorem}{Theorem}[section]
\newtheorem{df}[theorem]{Definition}
\newtheorem{propo}[theorem]{Proposition}
\newtheorem{lemma}[theorem]{Lemma}
\newtheorem{example}[theorem]{Example}
\newtheorem{remark}[theorem]{Remark}
\newtheorem{corollary}[theorem]{Corollary}

\newcommand*{\seq}[2][n]
            {{\ensuremath{#2_{1}, \allowbreak \ldots, \allowbreak #2_{#1}}}}
\newcommand*{\SEQ}[3]
            {{\ensuremath{#1_{#2}, \allowbreak \ldots, \allowbreak #1_{#3}}}}
\newcommand*{\SEQC}[3]
            {{\ensuremath{#1_{#2} \cdots #1_{#3}}}}

\newcommand*{\notmodels}{\mathrel{\,\not\!\models}}

\newcommand{\restrict}[2]{#1\,{\rule[-.65ex]{.3pt}{2.5ex}}\,_{ #2}}

\newcommand*{\vars}{{\ensuremath{\it vars}}}

\newcommand*{\HU}{{\ensuremath{\cal{H U}}}\xspace}
\newcommand*{\HB}{{\ensuremath{\cal{H B}}}\xspace}
\newcommand*{\TU}{{\ensuremath{\cal{T U}}}\xspace}
\newcommand*{\TB}{{\ensuremath{\cal{T B}}}\xspace}
\newcommand*{\M}{{\ensuremath{\cal M}}\xspace}
\newcommand*{\T}{{\ensuremath{\cal T}}\xspace}

\usepackage[mathscr]{euscript}

\newcommand*{\NN}{{\ensuremath{\mathbb{N}}}\xspace}

\newcommand*{\pre}{{\ensuremath{\mathit{pre}}}\xspace}
\newcommand*{\post}{{\ensuremath{\mathit{post}}}\xspace}

\newcommand*{\pgets}{\ensuremath{\mathrel{\mbox{\tt:-}}}}

\DeclareRobustCommand{\mydate}
   {\makebox[0pt][r]{\ifthenelse{\boolean{commentsaon}}{Draft, }{}%
                     \today}}

\title{Proving completeness of logic programs with the cut}
\author[W. Drabent]
  {W{\l}odzimierz Drabent 
\hfill \makebox[0pt][l]{\hspace{7.85em}\normalsize\mydate}
\\
    Institute of Computer Science,
         Polish Academy of Sciences,
         and
         IDA, Link\"opings universitet, Sweden
 \\     {\tt drabent\,{\it at}\/\,ipipan\,{\it dot}\/\,waw\,{\it dot}\/\,pl}
  }

\begin{document}
\maketitle

\begin{abstract}
Completeness of a logic program means that the program  
produces all the answers required by its specification.
The cut is an important construct of programming language Prolog.
It prunes part of the search space, this may result in a loss of completeness.
This paper proposes a way of proving completeness of programs with the cut.
The semantics of the cut is formalized by describing how SLD-trees are pruned.
A sufficient condition for completeness is presented, proved sound, and
illustrated by examples.

\end{abstract}


\section{Introduction}

Some constructs of programming language Prolog prune part of the search space,
i.e.\ of an LD-tree.
The basic pruning construct is the cut.
Pruning does not change the declarative meaning of a program; 
the program treated as a set of logic formulae is the same
with and without pruning constructs.
What is changed is the operational semantics -- the way the program is executed.
As pruning means skipping some fragments of the search space, it may result
in Prolog missing some answers. 
This paper presents a way of proving completeness of programs with the cut,
i.e.\ 
proving that a Prolog program would produce all the answers required by a
specification.

The work has to be based on a formal semantics.
Usually the semantics of the cut is described
in terms of explicit representation of computation
states, stacks of backtrack points, 
numerical labels related to cut invocations
etc, like in 
\cite{DBLP:journals/tcs/Billaud90,DBLP:journals/scp/Vink89,%
DBLP:journals/tplp/Schneider-KampGSST10,DBLP:journals/tplp/Andrews03}.
Some approaches require transforming programs into a special syntax
\cite{DBLP:journals/tcs/Billaud90,DBLP:conf/flops/KrienerK14}, or restrict
the class of programs dealt with (\cite{DBLP:conf/flops/KrienerK14}
requires so-called cut-stratification).
Some approaches describe only approximations of the semantics.
  The semantics of
 \cite{DBLP:journals/tplp/Schneider-KampGSST10}
does not distinguish success from failure, as the purpose of the semantics is
termination analysis.
The semantics of \cite{DBLP:conf/flops/KrienerK14}
may describe answers which actually are not computed; 
such inaccuracy is acceptable as the semantics is intended as a basis for
abstract interpretation, which introduces inaccuracies anyway.
Of course such semantics is inadequate for reasoning about program completeness.

In this paper we define the semantics of programs with the cut
in terms of pruning LD-trees.
Such approach is convenient -- the main proof of this paper is based on
comparing pruned and non-pruned LD-trees.
It is also closer to the usual way of describing operational semantics
(in terms of SLD-resolution)  than the approaches mentioned above.
Two approaches somewhat similar to ours
are those of Apt \cite{Apt-Prolog}
and Spoto \cite{DBLP:journals/jlp/Spoto00}.
They also employ trees, but the trees are not defined as subgraphs of the
LD-tree. 
In \cite{Apt-Prolog}, initial queries containing the cut seem not
considered. 
The approach of \cite{DBLP:journals/jlp/Spoto00}
seems more complicated than ours and that of  \cite{Apt-Prolog}.
For instance,
new cut symbols are added, to embed in tree nodes information about 
the origin of each cut.
Our formal semantics considers definite clause programs with the cut, and
Prolog selection rule.
Similarly to \cite{Apt-Prolog,DBLP:journals/jlp/Spoto00},
it does not deal with modifying the selection 
rule by means of delays (also called coroutining).
  Other control constructs, like
the conditional or negation as failure, can be expressed by means of the cut.

Little work has been done on reasoning about completeness of logic programs, see
\cite{drabent.arxiv.coco14} (or \cite{drabent.lopstr14}),
\cite{Deransart.Maluszynski93},  and the references therein;
see also Section \ref{sec:compl.decl} for a particular completeness
proving method.

An approach to proving completeness in presence of pruning is presented 
in  \cite{drabent.arxiv.coco14} (also reported in \cite{drabent.lopstr14}).
It is based on a more abstract view of pruning than this work.
It does not directly refer to a pruning construct in programs,
like the cut.  
For a completeness proof, 
one has to separately find out which clauses are applied to
each selected atom in a pruned SLD-tree.  
Determining the clauses may be not obvious, as a single invocation of the cut
may result in pruning children of many nodes of a tree.
  Moreover, 
  different numbers of children may be pruned 
  for nodes with similar selected atoms.
On the other hand, 
the approach is not restricted
to the selection rule of Prolog, and applies to any kind of pruning.
The author is not aware of any other work on proving properties of programs with
the cut, particularly on proving completeness.  

A related subject is abstract interpretation;
for its applications to programs with the cut
see \cite{DBLP:journals/tplp/Schneider-KampGSST10,DBLP:conf/flops/KrienerK14}
and the references therein.
In abstract interpretation, properties of programs are derived automatically,
however the class of possible properties is restricted to the chosen abstract
domain.

The main result of this paper is a sufficient condition for completeness of
LD-trees pruned due to the cut.  (As completeness depends on initial queries
we formally do not talk about program completeness, but completeness of
trees.)  The sufficient condition is proved sound w.r.t.\ the formal
semantics.  It is illustrated by a few examples.
A preliminary version of the sufficient condition, restricted to the cut in
the last clause of a procedure, appeared in \cite{drabent.lopstr14}.

Pruning constructs, like the cut, may destroy completeness of programs,
but they preserve program correctness.  However it is possible that 
a logic program is incorrect, but behaves correctly (for some initial
queries) under pruning, as wrong answers are pruned.
Such programming technique is called ``red cut''.  
Proving correctness in such case is outside of the scope of this paper,
and is a subject of future work.
See \cite{drabent.arxiv.coco14} 
for a sufficient condition for correctness in a context
of the other approach to pruning.
Another subject of future work is dealing with other selection rules
(Prolog with delays).

Let us outline the rest of the paper.
The next section is an overview of basic concepts. 
Section \ref{sec:operational.sem} deals with the operational semantics,
first discussing LD-resolution and then introducing the semantics of
LD-resolution with the cut.
Section \ref{sec:corr-compl} discusses proving correctness and completeness
of programs without pruning.
In particular, it discusses a specific notion of correctness
related to the operational semantics (LD-resolution); this notion is 
needed in the next section.
Section \ref{sec:compl.cut} presents a sufficient condition for completeness 
in the presence of the cut.
Section \ref{sec:examples} presents
some example proofs of program completeness.
The Appendix contains proofs missing in
Section \ref{sec:compl.cut}.

\section{Preliminaries}
In this paper we consider definite clause programs (informally -- logic
programs without negation),
 and Prolog programs that are definite clause programs with
(possibly) the cut.
We assume that the reader is familiar with basics of Prolog, and basics of
the theory of logic programming, including the notions of
(Herbrand) interpretation/model, logical consequence, (definite clause) program, query,
the least Herbrand model, substitution, unification,
SLD-derivation, SLD-tree, and soundness / completeness of SLD-resolution
\cite{nilsson.maluszynski.book,lloyd87,Apt-Prolog}.
We follow the definitions and notation of \cite{Apt-Prolog}, unless
stated otherwise.  
In particular,  the elements of
SLD-derivations and nodes of SLD-trees are queries, i.e.\ conjunctions of atoms,
represented as sequences of atoms.
(Instead of queries, 
the other approach \cite{nilsson.maluszynski.book,lloyd87} uses 
goals, i.e.\ negations of queries.)
LD-resolution (LD-derivation, LD-tree) is SLD-resolution (SLD-derivation,
SLD-tree)  with Prolog selection rule -- in any query its
first atom is selected; see also Section \ref{sec:LD}.

\pagebreak[3]

Following \cite{Apt-Prolog}, we assume that truth of a formula is defined in
such a way that  $I\models F$ iff $I\models\forall F$ (for any formula $F$
and any interpretation, or theory, $I$).
An atom whose predicate symbol is $p$ will be called
a {\em$p$-atom} (or an {\em atom for $p$}).  Similarly, a clause
whose head is a $p$-atom is a {\em clause for $p$}.  
In a program $P$, 
by {\em procedure $p$} we mean 
the set of clauses for $p$ in $P$.

We do not require that the considered alphabet consists only 
of the function and predicate symbols occurring in the considered program.
The Herbrand universe (i.e.\ the set of ground terms) will be denoted by \HU,
the Herbrand base (the set of ground atoms)  by \HB,
and the sets of all terms, respectively atoms, by \TU and \TB.
For an expression (a program) $E$
by $ground(E)$ we mean the set of ground instances
of $E$ (ground instances of the clauses of $E$).
$\M_P$ denotes the least Herbrand model of a program $P$.

By an {\em answer} for a program $P$ we mean a query $Q$ such that 
$P\models Q$.
(In \cite{Apt-Prolog} answers are called ``correct instances of queries''.)
By a {\em computed answer} for a program $P$ and a
   query $Q_0$ we mean an instance $Q\theta$ of $Q_0$ where $\theta$ is a
   computed answer substitution \cite{Apt-Prolog} 
   obtained from some successful SLD-derivation for $Q_0$ and $P$.
Often it is not necessary to distinguish answers and computed answers, as 
by soundness and completeness of SLD-resolution,
$Q$ is an answer for $P$ iff $Q$ is a computed answer for $P$ (and some query).

  Names of variables begin with an upper-case letter.
The set of variables occurring in an expression $E$ will be denoted 
$\vars(E)$. 
For a substitution \linebreak[3] $\theta = \{X_1/t_1,\ldots,X_n/t_n\}$, 
we denote
$do m(\theta)=  \{X_1,\ldots,X_n\}$, $r n g(\theta)=\vars(\{\seq t\})$,
and $\vars(\theta) = do m(\theta)\cup r n g(\theta)$.
The substitution $\theta$ is {\em ground} if 
$\seq t$ are ground terms.
Note that if $\theta,\sigma$ are ground substitutions with disjoint domains, 
$do m(\theta)\cap do m(\sigma) = \emptyset$,
then $\theta\sigma = \theta\cup\sigma = \sigma\theta $.
The restriction of
$\theta$ to a set $V$ of variables is
$ \restrict \theta V = \{\, X/t \in \theta \mid X\in V \,\} $.
By $\restrict \theta E$ we mean $\restrict \theta {\vars(E)}$.
The empty substitution is denoted by $\epsilon$.
A substitution $\theta$ is {\em idempotent} when
 $do m(\theta)\cap r n g(\theta) = \emptyset$.
Abbreviation mgu stands for ``most general unifier''.
A unifier $\theta$ of expressions $E_1, E_2$ is {\em relevant}
if $\vars(\theta)\subseteq\vars(E_1)\cup\vars(E_2)$.

  We use the list notation of Prolog.  So 
  $[\seq t]$  ($n\geq0$) stands for the list of elements $\seq t$.%
\footnote
{\newcommand{\mydot}{\mbox{{\LARGE\ensuremath{\bm.}\hspace{-1pt}}}}
Formally, $[\seq t]$ is an alternative notation for the term 
$
\mydot(t_1, \mydot(t_2, \ldots ,\mydot(t_n,[\,])\ldots))  $,
where {\LARGE\ensuremath{\bm.}}  is the list constructor and $[\,]$ is the
empty list.
Also, $[t|u]$, $[s,t|u]$ is an alternative notation for, respectively,
$\mydot(t,u)$, $\mydot( s, \mydot(t,u))$.
}
  Only a term of this form is considered a list.
 (Thus terms like $[a,a|X]$, or  $[a,a|a]$ where $a$ is distinct from $[\,]$,
  are not lists).
Sometimes, in examples, we will use the Prolog symbol \pgets\ 
instead of $\gets$ in programs.
The set of natural numbers will be denoted by \NN.

\section{Semantics for definite clause programs with the cut}
\label{sec:operational.sem}
This section formalizes a main part of the semantics of Prolog.
We present an operational semantics of definite clause programs augmented
with the cut ($!$).  First we abstract from the cut, describing LD-resolution.
Then we describe how the cuts prune LD-trees.
We begin with a note of declarative semantics.

To incorporate the cut into programs, let us add a new 0-argument predicate
symbol $!$ to the alphabet, and extend the set \TB of atoms:
{$\TB^+ = \TB\cup \{\,!\,\}$}.
A {\em program with cuts} is a finite sequence of definite clauses
of the form $H\gets\seq B$, where $n\geq0$, $H\in\TB$, and $\seq B\in\TB^+$.

In the rest of the paper we write ``program'' for ``program with cuts''.
Sequences of atoms from $\TB^+$ will often be denoted 
by $\vec A, \vec B$ etc, with possible indices. 
When this does not lead to ambiguity,
we sometimes treat queries as sets of atoms, and programs with cuts as sets
of clauses, and write e.g.
$\vec A\subseteq S$
to say that each atom of the sequence $\vec A$ is in the set $S$,
or say that a clause is a member of a program.

\paragraph{Declarative semantics.}
When considering programs from the point of view
of logic, atom ! will be treated as true in each interpretation.
Thus  $I\models \vec A, {!},\vec B$  iff $I\models \vec A,\vec B$
(where $\vec A, \vec B\subseteq\TB^+$).
So, in what follows we assume that interpretations do not describe the
semantics of~!. 
Hence by a Herbrand interpretation we mean a set of ground atoms from \HB.
Assume that a definite program $P'$ is a program with cuts $P$ 
with each ! removed.
Then $P,P'$ have the same models, the same Herbrand models, and thus
the same least Herbrand model and the same answers.

\subsection{LD-resolution.}
\label{sec:LD}
For our purposes we need a slight generalization of the standard
LD-resolution for programs with the cut.
The role of the cut is pruning LD-trees.  So we first consider LD-resolution,
where the cut is neglected, and then we introduce the semantics of the cut by
defining how LD-trees are pruned.

An {\bf LD-derivation} for a program $P$
is a pair of (finite or infinite) sequences: 
a sequence $Q_0,Q_1,\ldots$ of queries, 
and a sequence $\theta_1,\theta_2,\ldots$ of mgu's.
(The sequences are either both infinite, or both finite
with respectively $n+1,n$ elements, $n\geq0$.)
When $Q_{i-1}= {!}, \vec A$ then $Q_{i}= \vec A$ and $\theta_i=\epsilon$
(the empty substitution).  
Otherwise the successor of $Q_{i-1}$, if any, is as in the standard
LD-resolution:  When
$Q_{i-1}= A, \vec A$ then $Q_{i}= (\vec B,\vec A)\theta_i$, where
$\theta_i$ is an mgu of $A$ and $H$ and $H\gets\vec B$ is a standardized
apart variant%
\footnote{%
  This means that no variable of $H\gets\vec B$ occurs in $\SEQ Q 0 {i-1} $,
  $\seq[i-1]\theta$, or in a clause variant used in deriving
  some $Q_j$, for $0<j<i$.
  If (some variant of) the head of $C$ is unifiable with $A$ then we say that 
  $C$ is {\em applicable} to query $A, \vec A$.
}
 of a clause $C$ of the given program.
Without loss of generality we can assume that the employed mgu's are
idempotent and relevant.
A derivation $Q_0,\ldots,Q_n;\seq\theta$ is {\em successful} if its last
query is empty.
The (computed) {\em answer} of such derivation is $Q_0\theta_1\cdots\theta_n$.

As a query $Q$ may occur in a derivation $D$ a few times, one should speak about
occurrences of queries in derivations.  The same for an atom in a query, an
atom selected in a derivation, etc.
However, to simplify the presentation, we usually skip the word ``occurrence''.

The notion of derivation described above is slightly different from those of
\cite{Apt-Prolog} and \cite{Doets}.  
In \cite{Apt-Prolog} a proper prefix of a derivation is not a derivation,
while here it is.
In \cite{Doets} the substitutions of a derivation are not the mgu's, but
{\em specializations} (mgu's restricted to the variables of queries);
instead of $\theta_i$ there is a specialization $\restrict{\theta_i}{Q_{i-1}}$.

Consider a derivation $D$ containing a query  $Q_j = \vec B,\vec A $. 
We describe which fragment of $D$ corresponds to an evaluation of $\vec B$.

\begin{df}[subderivation]
Let $ D =Q_0,Q_1,\ldots;\theta_1,\theta_2,\ldots$ be a (finite or infinite)
LD-derivation, and $Q_j = \vec B,\vec A $ be a query in $D$. 
If $D$ contains a query $Q_m= \vec A \theta_{j+1}\cdots\theta_m$, where
$m\geq j$, then $\vec B$ (of $Q_j$) {\bf succeeds} in $D$.

If $\vec B$ (of $Q_j$) does not succeed in $D$ then 
the {\bf subderivation} of $D$ for $\vec B$ (of $Q_j$) is the
(finite or infinite) derivation
$D_j =Q_j,Q_{j+1},\ldots;\theta_{j+1},\theta_{j+2},\ldots$
that contains each query $Q_i$ and substitution $\theta_{i+1}$ of $D$
such that  $i\geq j$.

If $\vec B$ (of $Q_j$) succeeds in $D$ then the {\bf subderivation} of $D$
for $\vec B$ (of $Q_j$) is the derivation 
$D_j =Q_j,\ldots,Q_m;\theta_{j+1},\ldots,\theta_{m}$, where 
$Q_m= \vec A \theta_{j+1}\cdots\theta_m$ and, for $i=j,\ldots,m-1$,
$Q_i=\vec B_i,\vec A\theta_{j+1}\cdots\theta_i$ with nonempty $\vec B_i$.
Such subderivation is called {\bf successful}, and
 $\vec B\theta_{j+1}\cdots\theta_m$ is called the
(computed in $D$)  {\bf answer}   for $\vec B$.
  
\end{df}
A subderivation for an atom $p(\vec t)$ of a query $p(\vec t),\vec A$
within a derivation $D$
may be informally understood as a procedure invocation (of procedure $p$).
In an extreme case of empty $\vec B$ (i.e.\ $Q_j=\vec A$), the subderivation for
$\vec B$
of $Q_j$ consists of a single query $Q_j$ (and no substitutions).
Due to the clauses being standardized apart and the mgu's being relevant, we have:

\begin{lemma}
\label{lemma:subderivation-variable}
Let $D$, $Q_j$ and $D_j$ be as in the definition above
($D$ an LD-derivation,  $Q_j=\vec B,\vec A$ a query of $D$, 
and $D_j$ be the subderivation for $\vec B$ starting at $Q_j$).
Assume that a variable $X$ occurs in $Q_0,\ldots,Q_j;\seq[j]\theta$
or in a clause variant used to derive some of $\seq[j]Q$, and that 
$X$ does not occur in $\vec B$.
Then $X$ does not occur in any mgu $\theta_{j+1},\theta_{j+2},\ldots$ of $D_j$.
Neither it occurs in the prefix $\vec B_i$ of any query 
$Q_i=\vec B_i,\vec A\theta_{j+1}\cdots\theta_i$  of $D_j$ (for $i>j$).
\end{lemma}

\begin{proof}
  By induction on $i$
(as the clauses employed in $D$ are standardized apart, and the mgu's are
 relevant). 
\end{proof}

\pagebreak[3]
The {\bf LD-tree} for a program $P$ and a query $Q$ is defined in a standard way.
The root of the tree is $Q$ and its branches are LD-derivations.
A node $Q'$ to which $k$ clauses of $P$ are applicable, has $k$ children, one
for each such clause of $P$.  The ordering of the children follows that of
the clauses in $P$.
See also Example \ref{ex:pruning:tree} below.
Formally, LD-trees are trees with nodes labelled with queries.
We will often simply say that a node is a query, taking care that this does
not lead to ambiguities.

\subsection{Semantics of the cut.}
LD-trees are the search spaces of Prolog computations.  The role of the cut 
is to skip the search of some fragments of an LD-tree.
We first formalize the order in which Prolog searches LD-trees.
Note that the part of the tree to the right of (any) infinite path is not
searched.

\begin{df}
The preorder sequence $seq(T)$ of the nodes of (an ordered) tree $T$ is defined
recursively as
\[
seq(T) \ = \ Q,seq(T_1),\ldots,seq(T_i),
\]
where 
\begin{tabular}[t]{l}
    $Q$ is the root of $T$, with the children \seq Q (in this order), 
    $0\leq i\leq n$,
\\
    $\seq T$ are the subtrees of $T$ rooted in \seq Q, respectively,
\\
    $\seq[i-1]T$ are finite,
\\
    $i=n$ or $T_i$ is infinite.
\end{tabular}
\end{df}

Before introducing the formal semantics of the cut, we describe 
the semantics informally and illustrate it by an example.
Consider an LD-tree $T$.
When Prolog visits a node $Q''$ with the cut selected,
some nodes of $T$ are pruned; in other words, they will not be
visited in the further search of $T$.  The pruned nodes are (some)
descendants of $Q'$ --
the node which introduced the cut of $Q''$. (The cut appeared first in
a child of $Q'$.)
All the descendants to the right of the path $Q',\ldots Q''$ are pruned.

\newcommand{\arca}{
                   \begin{pgfpicture}
                         \pgfpathmoveto{\pgfpointorigin}
                         \pgfpathlineto{\pgfpoint{1ex}{2ex}}
                         \pgfusepath{stroke}
                   \end{pgfpicture}
                  }
\newcommand{\arcb}{
                   \begin{pgfpicture}
                         \pgfpathmoveto{\pgfpointorigin}
                         \pgfpathlineto{\pgfpoint{-1ex}{2ex}}
                         \pgfusepath{stroke}
                   \end{pgfpicture}
                  }
\newcommand{\arcaa}{
                   \begin{pgfpicture}
                         \pgfpathmoveto{\pgfpointorigin}
                         \pgfpathlineto{\pgfpoint{4ex}{1.7ex}}
                         \pgfusepath{stroke}
                   \end{pgfpicture}
                  }
\newcommand{\arcbb}{
                   \begin{pgfpicture}
                         \pgfpathmoveto{\pgfpointorigin}
                         \pgfpathlineto{\pgfpoint{-4ex}{1.7ex}}
                         \pgfusepath{stroke}
                   \end{pgfpicture}
                  }
\newcommand{\arcv}{\rule{.4pt}{1.7ex}}

\noindent
\begin{minipage}[t]{.5\textwidth}
\begin{example}
\label{ex:pruning:tree}
Consider the program and the LD-tree from the diagram.
The cut executed in node ${!},r,{!}$ prunes the descendants of the nodes
of the path  $ q,{!};\ s, {!},r,{!};\ {!},r,{!}$ to the right of the path.
So after visiting the node ${!},r,{!}$ Prolog visits nodes $r,{!}$ and $r$.

\quad
Assume that the rule $q \pgets s, {!}, r$ is removed from the program.
Hence node  $s, {!},r,{!}$  and  its descendants are removed from the tree.
Now the cut in node ${!}$ is executed.  This prunes the nodes to the right of
the path  $p;\  q,{!};\ t,{!};\ {!}$, namely $r,r$ and $r$.

\end{example}
\end{minipage}
\hfill
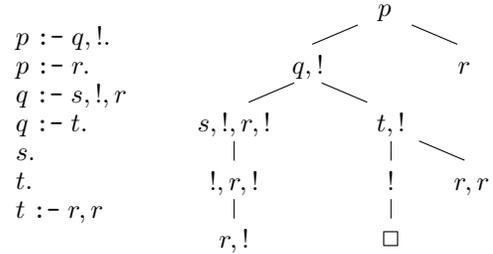
\begin{minipage}[t]{.45\textwidth}

\vspace{\parsep}
\mbox{}\ 
$
  \begin{array}[t]{l}
    \\
    p \pgets q,{!}.  \\
    p \pgets r.      \\
    q \pgets s, {!}, r  \\
    q \pgets t.      \\
    s.               \\
    t.               \\
    t \pgets r,r
  \end{array}
$
\hfil\hfil
%
%
%
%
%
%
\input{diagram.v2}

\hfill\mbox{}
\end{minipage}

\medskip\smallskip

The next two definitions describe the semantics of the cut.
\begin{df}
\label{def:cutting.sequence}
Consider an LD-tree $T$ with a branch $D$ containing consecutive nodes
\[
\begin{array}{l}
 Q_{j-1}= A,\vec A \\ 
 Q_{j}= (\vec B_1,{!},\vec B_2,\vec A)\theta_j \\
 \cdots \\
 Q_{m}= ({!},\vec B_2,\vec A)\theta_j\cdots\theta_m \\
\end{array}
\]
such that $j\leq m$, and
$Q_j,\ldots,Q_m;\SEQ\theta j m$ is a subderivation of $D$ 
for $\vec B_1$.%
\footnote
{Thus a clause variant $H\gets \vec B_1,{!},\vec B_2$ was applied to $Q_{j-1}$,
   and, for $i=j,\ldots,m-1$, 
   each $Q_i$ is of the form 
   $(\vec A_i,{!},\vec B_2,\vec A)\theta_j\cdots\theta_i$ with nonempty
   $\vec A_i$.
}
We say that the node $Q_{j-1}$ {\bf introduces} the cut of  $Q_{j}$, 
and that the cut of $Q_{j}$ is {\bf potentially executed} in the node $Q_m$.%
\footnote{%
We write ``the cut of $Q_{j}$'' for brevity.
Formally we deal here with the occurrence of $!$ 
between $\vec B_1\theta_j$ and $\vec B_2\theta_j$ in the node $Q_j$.
When there are more such occurrences then 
the objects introduced by this definition are defined separately for each
occurrence of ! in $Q_j$.
}
Derivation $Q_{j-1},\ldots,Q_m$ is called a {\bf cutting sequence} of nodes
in $T$.  Its first node $Q_{j-1}$ will be called the {\bf introducing node},
and its last node $Q_m$ -- the {\bf executing node} of the
cutting sequence. 

For a case where the cut occurs in the initial query ($j=0$,
$\theta_0=\epsilon$, and $Q_{j-1}$ does not exist),
``potentially executed'' is defined as above, and 
the cutting sequence is $\SEQ Q 0 m$.
(Note that such cutting sequence does not have its introducing node.)

Each node of the form ${!},\vec A$ in the tree is the executing node of a unique cutting sequence; the sequence will be called the 
{\bf cutting sequence of the executing node} ${!},\vec A$.

The nodes of $T$ {\bf pruned by a cutting sequence} $Q_{j-1},\ldots,Q_m$
are those children of each $Q_{i-1}$ that are to the right of $Q_{i}$,
for $i=j,\ldots,m$, and the descendants of the children.
The nodes {\bf pruned by an executing node}
 $Q_m$ are the nodes pruned
by the cutting sequence of $Q_m$. 
\end{df}

\noindent
Note that if a node $Q$ is pruned by a cutting sequence $D$ then
$Q$ does not precede in $seq(T)$ any of the nodes of $D$.
(More precisely, if $Q$ and the nodes of $D$ occur in $seq(T)$ then each node
of $D$ precedes $Q$.)
In  \cite{Apt-Prolog},
the introducing node of (the cutting sequence of) an executing node ${!},\vec
A$ is called the origin of the cut atom in ${!},\vec A$.

The definition above
describes pruning due to a single cut in an LD-tree 
(or -- more generally -- pruning trees in which no executing node is pruned).
When there are more cuts, an executing node $Q'$ may prune an executing node
$Q''$.  Moreover, Ex.\,\ref{ex:pruning:tree} shows that in some cases
some nodes pruned by $Q''$ are not pruned by $Q'$ and remain in the final
tree. 
This should be considered while defining pruned LD-trees

For the next definition we need to consider trees which are subgraphs of
LD-trees. 
By a cutting sequence of a subgraph $T'$ of an LD-tree $T$ we mean a cutting
sequence $D$ of $T$, such that each node of $D$ is in $T'$.
By an executing node $Q'$ of $T'$ we mean a node of the form ${!},\vec A$.

\begin{df}[pruned LD-tree]

Let $T$ be an LD-tree.
Consider the possibly infinite sequence \T of trees
$ T_0, T_1,\ldots$,
such that $T_0=T$ and 
\begin{itemize}
\item 
$T_i$ is obtained from $T_{i-1}$ by removing the nodes pruned by
 the $i$-th executing node in $seq(T_{i-1})$
(for each $T_i$ in \T, where $i>0$),

\item
if some $seq(T_n)$ (where $n\geq0$) contains exactly $n$ executing nodes then the sequence
\T is finite and $T_n$ is its last element, otherwise \T is infinite.
\end{itemize}
Let $T'$ be 
the subgraph of $T$ containing the nodes occurring in each of the trees
$T_0,T_1,\ldots$.
(Thus when the sequence is finite then $T'$ is its last element.)
The {\bf pruned LD-tree} $pruned(T)$ corresponding to $T$ 
(shortly: the pruned $T$)
consists of those nodes of $T'$ that occur in $seq(T')$.
\end{df}

For informal explanation of the definition, consider the subgraph
 $T''$ of $T_{i-1}$
consisting of those nodes
that are in  $seq(T_{i-1})$  between the root of $T_{i-1}$ and $Q_i$,
the $i$-th executing node in $seq(T_{i-1})$.
This subgraph describes the computation up to the $i$-th execution of the cut.
The nodes pruned by (this execution of) the cut are absent from $T_i$.
Whole $T''$ is a subgraph of $T_i$.
Also,  $T''$ will remain unchanged in (i.e.\ be a subgraph of)
all the subsequent trees $T_i,T_{i+1},\ldots$.

\section{Correctness and completeness of programs}
\label{sec:corr-compl}
In preparation for the main subject of this work --
program completeness related to the operational semantics with pruning,
this section discusses some semantic issues abstracting from pruning.
The purpose is twofold,
introducing some concepts needed later on,
and providing a ground for comparing 
the proof methods based on declarative semantics, with the method of this paper,
dealing with pruning.
First we discuss correctness and completeness of programs,
two notions related to the declarative semantics.
We also present the standard ways of reasoning about program termination.
Then we discuss a specific notion of correctness related to operational
semantics, namely to LD-resolution.

\subsection{Declarative notions of correctness and completeness}
\label{sec:compl.decl}

\paragraph{Specifications.}
From a declarative point of view, logic programs compute relations. 
A specification should describe these relations.
It is convenient to assume that the relations are over the Herbrand universe. 
A handy way for describing such relations is a Herbrand interpretation;
it describes, as needed, a relation for each predicate symbol of the program.
So, by a {\bf specification} we mean a Herbrand interpretation, i.e.\ a
subset of \HB.

\paragraph{Correctness and completeness.}
In imperative and functional programming, (partial) correctness usually means
that the program results are as specified (provided the program terminates).
Logic programming involves non-determinism of a specific kind.  A query may
have 0, 1, or more answers, and the idea is to compute all of them.
Thus in logic programming the notion of correctness divides 
into two: {\em correctness} (all the results are
compatible with the specification) and {\em completeness} (all the results
required by the specification are produced). 
In other words, correctness means that the relations defined by the program are
subsets of the specified ones, and completeness means inclusion in the
opposite  direction. 
Formally:
\begin{df}
\label{def:corr:compl}
Let $P$ be a program and $S\subseteq\HB$ a specification.
$P$ is {\bf correct} w.r.t.\ $S$ when $\M_P\subseteq S$;
it is {\bf complete} w.r.t.\ $S$ when $\M_P\supseteq S$.
\end{df}
We will sometimes skip the specification when it is clear from the context.
It is important to understand the relation between specifications and
the answers of correct (or complete) programs
\cite{drabent.arxiv.coco14}.
A program $P$ is correct w.r.t.\ a specification $S$ iff (for any query $Q$)
$Q$ being an answer of $P$  implies $S\models Q$.
(Remember that $Q$ is an answer of $P$ iff $P\models Q$.)
Program $P$ is complete w.r.t.\ $S$ iff
$S\models Q$ implies that $Q$ is an answer of $P$ (for any ground query $Q$).%
\footnote{%
\label{footnote:example:notequivalent}%
To show that the equivalence does not hold for all queries,
assume a two element alphabet of
    function symbols, with a unary $f$ and
    a constant $a$.  Take $P = \{\, p(a).\ p(f(X)).\}$, $S=\HB$, 
    $Q=p(X)$. 
The program is complete w.r.t.\ $S$ and $S\models Q$, but $Q$ is not an
answer of $P$.
}  %
For arbitrary queries, completeness of $P$ w.r.t.\ $S$ implies $P\models Q$
when in the underlying alphabet
there is a non-constant function symbol not occurring in $P,Q$,
or there are $k$ constants not occurring in $P,Q$, where $k\geq0$
is the number of distinct variables occurring in $Q$
\cite{drabent.arxiv.coco14}.
In particular, the implication holds when the alphabet of function symbols is
infinite (and $P$ is finite)
   \cite{DBLP:books/mk/minker88/Maher88}.
(See \cite{drabent.Herbrand.arxiv.2015} for further discussion.)

A note on pragmatics of the notion of completeness may be useful.
Remember that
the relations described by a specifications are on {\em ground}
terms.  
So, strictly speaking, specifications do not describe program answers, but
ground instances of the answers.
For a non-ground $Q$,
it depends on the underlying alphabet whether $S\models Q$ or not.%
\footnote{%
In the example from footnote \ref{footnote:example:notequivalent}
we have $S\models Q$, but when a new
function symbol is added to the alphabet then $S\notmodels Q$.
}
Informally, obtaining an answer $A\in\HB$ from a computation (an SLD-tree)
means that $A$ is a ground instance of a computed answer of the tree.
We are not interested
whether $A$ is actually a computed answer, or a more general computed answer
has been produced.  Similarly, obtaining 
answers $A,A'\in\HB$ may happen when 
both of them are instances of a single computed answer, 
or they are (instances of) different computed ones.

\paragraph{Approximate specifications.}
It happens quite often in practice that the relations defined by a program
are not known exactly and, moreover, such knowledge is unnecessary.
It is sufficient to specify the program's semantics approximately.
More formally, to provide distinct specifications, say
$S_{\it c o m p l}$ and $S_{corr}$, for completeness and correctness.
The intention is that  $S_{\it c o m p l}\subseteq \M_P\subseteq S_{corr}$,
where $\M_P$ is the least Herbrand model of the program.
So the specification for completeness says what the program has to compute,
and the specification for correctness -- what it may compute.
In other words, the program should not produce any answers rejected by the
specification for correctness. 
It is irrelevant whether atoms from 
$S_{corr}\setminus S_{\it c o m p l}$ are, or are not, answers of the program.
Various versions of the program may have different semantics,
but each version should be correct w.r.t.\ $S_{corr}$ and complete w.r.t.\ 
$S_{\it c o m p l}$.
As an example, consider the standard append program, and atom 
$A = {\it append}([a], 1, [a|1])$.  It is irrelevant
whether $A$ is an answer of the program, or not. 
For further discussion and examples see 
\cite{drabent.arxiv.coco14,DBLP:journals/tplp/DrabentM05shorter},
see also Ex.\,\ref{ex:append:proof}.

\paragraph{Reasoning about correctness.}
Although it is outside of the scope of this paper, we briefly mention proving
program correctness.
A sufficient condition for a program $P$ being correct w.r.t.\ a
specification $S$ is $S\models P$.
In other words,
     for each ground instance 
    $
    H\gets \seq B
    $
    of a clause of $P$,
     if $\seq B\in S\cup\{\,!\,\}$ then $H\in S$.
Deransart \cite{DBLP:journals/tcs/Deransart93} attributes this method to
\cite{Clark79}.   
See  \cite{drabent.arxiv.coco14,DBLP:journals/tplp/DrabentM05shorter}
for examples and discussion.

\paragraph{Reasoning about completeness.}
Little work has been devoted to reasoning about completeness of programs.
See \cite{drabent.arxiv.coco14} for an overview.  We summarize the approach
from \cite{drabent.arxiv.coco14}, also presented in 
\cite{drabent.lopstr14}.
That approach is a starting point for the method introduced in this paper.
We first need two auxiliary notions.
\begin{df}
    A program $P$ is {\bf complete for a query} $Q$ w.r.t.\  $S$
    when
    $S\models Q\theta$ implies that $Q\theta$ is an answer for $P$,
    for any ground instance $Q\theta$ of $Q$.
\end{df}
Informally,  $P$ is complete for $Q$
when 
all the answers for $Q$ required by the specification $S$ are answers of $P$.
Note that a program is complete w.r.t.\ $S$
 iff it is complete w.r.t.\ $S$ for any query
 iff it is complete w.r.t.\ $S$ for any query $A\in S$.

\begin{df}
A program $P$ is {\bf semi-complete} 
w.r.t.\ a specification $S$ if
$P$ is complete w.r.t.\ $S$ for any query $Q$ for which there exists a finite SLD-tree.
\end{df}

Less formally, the existence of a finite SLD-tree means
that, under some selection rule, the computation for $Q$ and $P$ terminates.
(Sometimes this is called ``universal termination''.)
For a semi-complete program $P$,
if the computation
terminates then all the answers for $Q$ required by the specification have
been obtained.
In other words, $P$ is complete for query $Q$.
So establishing completeness may be done in two steps: showing
semi-completeness and termination. 
Obviously, a complete program is semi-complete.

Our sufficient condition for semi-completeness employs the
following notion, stemming from \cite{Shapiro.book}.
\begin{df}
\label{def:covered}
  A ground atom $H$ is
  {\bf covered by a clause} $C$ w.r.t.\ a specification $S$
  if $H$ is the head of a ground instance  
  $
  H\gets \seq B
  $
  ($n\geq0$) of $C$, such that all the atoms $\seq B$ are in $S\cup\{\,!\,\}$.

  A ground atom $H$ is {\bf covered} {\bf by a program} $P$ w.r.t.\ $S$
  if it is covered w.r.t.\ $S$ by some clause $C\in P$.
\end{df}

\begin{theorem}[semi-completeness {\rm\cite{drabent.arxiv.coco14}}]
\label{th:semi-complete}
If all the atoms from a specification $S$ are covered w.r.t.~$S$
 by a program $P$ 
then $P$ is semi-complete w.r.t.~$S$.
\end{theorem}

\begin{example}
\label{ex:append:proof}
  Consider the well-known program APPEND:
\[
  app(\,[\,],L,L\,). \qquad\qquad
  app(\,[H|K],L,[H|M]\,) \gets app(\,K,L,M\,). \ \ \ \ 
\]  
and a specification 
\[
S_{\rm APPEND}^0 = \{\, app(k,l,m)\in\HB \mid
                k,l,m \mbox{ are lists, } k*l=m
\,\},
\] 
where $k*l$ stands for the concatenation of lists $k,l$.
Consider an atom $A = app(k,l,m)\in S_{\rm APPEND}^0$.
If $k=[\,]$ then $A = app([\,],l,l)$ and $A$ is covered by the first clause
of the program.  Otherwise $A = app([h|k'],l,[h|m'])$, where $k'*l=m'$.
Thus $A$ is covered by an instance $app([h|k'],l,[h|m'])\gets app(k',l,m')$
of the second clause of APPEND.  Hence by Th.\,\ref{th:semi-complete}
APPEND is semi-complete w.r.t.\ $S_{\rm APPEND}^0$.
It is also complete w.r.t.\ $S_{\rm APPEND}^0$, 
as it terminates for each
$A \in S_{\rm APPEND}^0$
(we skip a simple proof \cite{drabent.arxiv.coco14} that the program
is recurrent, see Th.\,\ref{th:termination} below).
Note that APPEND is not correct w.r.t.\ $S_{\rm APPEND}^0$,
as it has answers whose some arguments are not lists, 
e.g.\ $app([a], 1, [a|1])$
(See \cite{drabent.arxiv.coco14} for specifications, w.r.t.\ which the
program is correct.)

\end{example}

\subsection{Reasoning about termination.}
\label{sec:termination}
Termination -- this means finiteness of (S)LD-trees -- is needed to
conclude completeness from semi-completeness,  
and will also be needed for the main result of this paper.
We now we briefly summarize basic approaches to proving program termination 
\cite{Apt-Prolog}.

\begin{df}
A {\em level mapping} is a
function $|\ |\colon \HB\to \NN$ assigning natural numbers to ground atoms.
We additionally assume that $|!|=0$.

A program $P$ is  {\bf recurrent} {w.r.t.\ a level mapping}~$|\ |$
\cite{DBLP:journals/jlp/Bezem93,Apt-Prolog} if, in
every ground instance  $H\gets\seq B\in ground(P)$ of its clause ($n\geq0$),
$|H|>|B_i|$ for all $i=1,\ldots,n$.
A program is {\em recurrent}
if it is recurrent w.r.t.\ some level mapping.   

\label{def:acceptable}%

A program $P$ is {\bf acceptable} w.r.t.\ a specification $S$ and a level
mapping $|\ |$ if 
$P$ is correct w.r.t.\ $S$, and for every
$H\gets\seq B\in ground(P)$
we have $|H|>|B_i|$ whenever $S\models B_1,\ldots,B_{i-1}$.
A program is {\em acceptable} if it is acceptable w.r.t.\ some level mapping
and some specification.

A query $Q$ is {\bf bounded} w.r.t.\ a level mapping  $|\ |$ if, 
for some $k\in\NN$, $|A|<k$ for each ground instance $A$ of an atom of $Q$.

\end{df}
The definition of acceptable is more general than that of 
\cite{AP93,Apt-Prolog},
which additionally requires $S$ to be a model of $P$.
 Both definitions make the same programs acceptable \cite{drabent.arxiv.coco14}.

\begin{theorem}
[termination {\rm\cite{DBLP:journals/jlp/Bezem93,AP93}}]
\label{th:termination}
If $P$ is a recurrent program and $Q$ a bounded query then all
SLD-derivations for $P$ and $Q$ are finite.
   
If a program $P$ is acceptable w.r.t. some specification and some level mapping
then all LD-derivations for $P$ and a bounded query $Q$ are finite.

\end{theorem}
Hence each SLD-tree for $P,Q$ in the first case, 
and the LD-tree for $P,Q$ in the second case
is finite (as programs are finite).
The second part of the theorem holds for a more general class of queries 
(bounded w.r.t.\ $S$) \cite{Apt-Prolog}; we skip the details.
It follows
that  when
    a (finite) program
    $P$ is (i)~semi-complete w.r.t.\ a specification $S$
    and (ii)~recurrent or acceptable w.r.t.\ some level mapping
    (and some specification $S'$) then $P$ is complete w.r.t.~$S$.

\subsection{A notion of correctness related to operational semantics}
\label{sec:corr.operational}
The subject of this paper is completeness of programs when the search space
is pruned by means of the cut.
Such operational semantics does not preserve some basic properties of
SLD-resolution.  For example an instance of a query $Q$ may succeed while
$Q$ fails 
(e.g.\ consider program $p(a)\pgets {!},q.;\ p(X).$, for which query $p(X)$
fails and $p(b)$ succeeds).
Also, we need to reason about the form of atoms selected in derivations.
So a declarative approach is no longer possible; we have to reason in
terms of the operational semantics,
in other words, to express and prove properties inexpressible in terms of
specifications, correctness, and completeness of 
Section \ref{sec:compl.decl}

This section presents a method of reasoning about the form of selected atoms
in LD-derivations, and the form of the corresponding successes.
The approach stems from \cite{DM88} and is due to \cite
{DBLP:conf/tapsoft/BossiC89}, we follow the presentation of \cite{Apt-Prolog}.
Specifications of another kind are needed here, let us call them 
c-s-specifications (c-s for call-success).

\begin{df}[c-s-correctness]

A {\bf c-s-specification} is a pair $\pre,\post$ of sets of atoms, closed under
substitution.
The sets  $\pre, \post \in\TB$ are called, respectively, {\em precondition} and 
{\em postcondition}.

A program is {\bf c-s-correct} w.r.t.\ a c-s-specification \pre,\post when
in each LD-derivation $D$
every selected atom is in $\pre\cup\{\,!\,\}$,
and each atomic computed answer (of a successful subderivation of $D$)
 is in $\post\cup\{\,!\,\}$,
provided that $D$ begins with an atomic query from $\pre$.

\end{df}
For c-s-correct programs and more general initial queries, see below.
The notion of c-s-correctness will be employed in the main part of this work.

\begin{df}[well-asserted]
Let $\pre,\post$ be a c-s-specification.
A clause $C$ is {\bf well-asserted} (w.r.t.\ $\pre,\post$) if for each 
(possibly non-ground) instance $H\gets \seq B$ of $C$ ($n\geq0$) 
\[
\begin{array}{l}
   \mbox{if }\ H\in\pre, \seq[k]B\in \post\cup\{\,!\,\}\ 
   \mbox{ then }\ B_{k+1}\in\pre\cup\{\,!\,\},\
   \mbox{ for } k=0,\ldots,n-1
   \\
   \mbox{if }\ H\in\pre, \seq B\in \post\cup\{\,!\,\}\ \mbox{ then }\ H\in\post.
\end{array}
\]

A program is well-asserted if every its clause is.

A query $Q$ is well-asserted (w.r.t.\ $\pre,\post$) when the
clause $p\gets Q$ is well-asserted w.r.t.\ 
$\pre\cup\{p\}, \linebreak[3]\post\cup\{p\}$,
where $p\in\HB$ is a predicate symbol not occurring in $P,\pre,\post$.
\end{df}
Note that the first atom of a well-asserted query is in $\pre\cup\{\,!\,\}$,
and that if all atoms of a query $Q$ are in $\pre\cup\{\,!\,\}$ then $Q$ is well-asserted.

The following sufficient condition follows from Corollaries 8.8 and 8.9 of
\cite{Apt-Prolog} (with an obvious generalization to programs with the cut).

\begin{theorem}[c-s-correctness]
\label{th:c-s-correctness}
Let $P$ be a program and $\pre,\post$ a c-s-specification.
If $P$ is well asserted w.r.t.\ $\pre,\post$
then $P$ is c-s-correct w.r.t.\ $\pre,\post$

\end{theorem}

   The definition of c-s-correctness involves only atomic initial queries.
For general queries, consider
 a c-s-specification $\pre,\post$, a program $P$, and a query $Q$.
If $P$ is c-s-correct, and $Q$ is well-asserted then 
in each LD-derivation $D$ for $P$
and $Q$ each selected atom is in $\pre\cup\{\,!\,\}$ and each atomic computed
answer (of a subderivation) is in   $\post\cup\{\,!\,\}$. 
More generally, each atom of a computed answer of a subderivation of $D$
is in $\post\cup\{\,!\,\}$, as \post is closed under substitution.

\section{Completeness in the presence of the cut}
\label{sec:compl.cut}
This section introduces a sufficient condition for
completeness of pruned LD-trees.
The main result is preceded by some necessary definitions.

\begin{df}
Let $T$ be an LD-tree, or a pruned LD-tree, and $Q$ be its root.
An {\em answer} of $T$ is the computed answer of a
successful LD-derivation which is a branch of $T$.

The tree $T$ 
is {\bf complete} w.r.t. a specification $S\subseteq\HB$ if, for any ground  $Q\theta$,
$S\models Q\theta$ implies that $Q\theta$ is an instance of an answer of $T$.
\end{df}
Informally, $T$ is complete iff it produces all the answers for
its root which are required by $S$.

The next definition is, in a sense, the main part of our sufficient condition
for completeness.  The idea is to require not only that a ground atom $A$ is
covered by a clause $C$, but also that
the tree node introduced by
$C$ cannot be pruned by a cut in a preceding clause.
Moreover, when the cut is present in $C$, say $C= H\gets\vec B_0,!,\vec B_1$,
then $A$ should be produced by $C$ employing an arbitrary answer to the
fragment  $\vec B_0$ of $C$.
To formalize this idea it is necessary to employ a c-s-specification,
to describe the atoms possibly selected in the LD-trees
and the corresponding computed answers.

\pagebreak[3]
\begin{df}
\label{def:c-covered}
Let $S\subseteq\HB$ be a specification, and  $\pre,\post$ a c-s-specification
($\pre,\post\in\TB$).
  A ground atom $A$ is {\bf c-covered} (contextually covered) 
w.r.t.\ $S$ and $\pre,\post$
by a clause $C$ occurring in a program $P$~%
if
\linebreak[3]
    \begin{enumerate}
      \item
        $A$ is covered by $C$ w.r.t.\ $S$, and
    \item
      \label{def:c-covered:cond:previous}
      if $C = H\gets\ldots$
      is preceded in $P$ by a clause $C'= H'\gets\vec A_0,!,\vec A_1$, 
      where both $H,H'$ have the same predicate symbol, 
      and ! does not occur in $\vec A_0$, then
       \begin{itemize}
       \item 
         for any atom $H''\in \pre$  such that $A$ is an instance of $H''$
        \item
          no ground instance $H''\theta$ of $H''$ is covered by 
          $H'\gets\vec A_0$ w.r.t.\ $post\cap\HB$;
       \end{itemize}

    \item
      \label{def:c-covered:cond:cut}
      if $C$ contains the cut, $C= H\gets\vec B_0,!,\vec B_1$, then

    \begin{itemize}
      \item 
        for any instance $H\rho\in \pre$ such that $A$ is an instance of $H\rho$
        (and $\rho$ is as below),

      \item 
        for any ground instance
        $\vec B_0\rho\eta$ such that $\vec B_0\rho\eta\subseteq post\cup\{!\}$
        (and $\eta$ is as below),
\ifthenelse{\boolean{commentsaon}}{\enlargethispage*{\baselineskip}}{}

    \item
      $A$ is covered by $(H \mathop\gets \vec B_1)\rho\eta$ w.r.t.\ $S$,
\nopagebreak
   \end{itemize}
\nopagebreak
where
 $do m(\rho)\subseteq \vars(H)$,
\,$r n g(\rho) \cap \vars(C) \subseteq \vars(H)$, 
 $do m(\rho)\cap r n g(\rho) = \emptyset$,
 and $do m(\eta) = \vars(\vec B_0\rho)$.
    \end{enumerate}
\pagebreak[3]
We say that $A$ is c-covered (w.r.t.\ $S$ and $\pre,\post$)
 by a program $P$ if it is c-covered (w.r.t.\ $S$ and $\pre,\post$)
by a clause from $P$.
Similarly, $S$ is c-covered by $P$ if each atom from $S$ is c-covered by $P$.
\end{df}

Some informal explanation may be useful.
The role of condition \ref{def:c-covered:cond:previous} is to exclude cases
where for query $H''$ the cut in a clause $C'$ preceding $C$ is executed, 
which results in not applying  $C$  for $H''$.
Roughly speaking, the cut in  $C'= H'\gets\vec A_0,!,\vec A_1$
is executed when $\vec A_o$ succeeds.
What we know about the computed answer for $\vec A_0$ obtained at the success
is that each atom of the answer is in \post.
So the cut in $C'$ may be executed if there is an instance
$(H'\gets\vec A_0)\varphi$, its head $H'\varphi$ is an instance of $H''$
and $\vec A_0\varphi\subseteq\post$.  It is sufficient here to consider
only ground instances of $H'\gets\vec A_0$;  such a ground instance exists 
iff a ground instance of $H''$ is covered by  $H'\gets\vec A_0$ w.r.t.\ 
$\post\cap\HB$.
When the cut is executed in clause $C= H\gets\vec B_0,!,\vec B_1$
then only the first answer for $\vec B_0$ will be used.  
The only information we have about this answer
 is that its atoms are in $post\cup\{!\}$.
The role of condition \ref{def:c-covered:cond:cut} is to assure that for each 
such answer clause $C$ can produce $A$.
Tu assure that such answer exists, $C$ is required to cover $A$ w.r.t.\ $S$.

For programs without the cut, c-covered is equivalent to covered.
For multiple occurrences of the cut in a clause,
condition \ref{def:c-covered:cond:previous} considers the first
occurrence of $!$ in $C'$, while
what matters in condition \ref{def:c-covered:cond:cut} 
is the last occurrence of $!$ in $C$
(if the condition holds for the last occurrence then it holds for each
previous one).
The two conditions get simplified when all the atoms of $\pre$ are ground,
as then $H''=A=H''\theta$ and $H\rho = A$.  
In a general case,
checking that an atom is c-covered by a clause can be simplified as follows:

\pagebreak[3]
\begin{remark}
\label{remark:conditions}
Note that in condition \ref{def:c-covered:cond:previous}
of Def.\,\ref{def:c-covered},
instead of considering all atoms $H''\in\pre$,
it is sufficient to
consider maximally general atoms $H''\in\pre$ unifiable with $H'$
(and having $A$ as an instance).
\nopagebreak

Similarly,
by Proposition \ref{lemma:instances:c-covered:tough:condition} in the Appendix,
instead of considering all instances $H\rho\in\pre$ of $H$
in condition \ref{def:c-covered:cond:cut}, it is sufficient to consider
maximally general instances $H\rho\in\pre$.
\nopagebreak

Assume that $S\subseteq \post$.
Then, for $A$ to be is covered by a clause $C= H\gets\vec B_0,!,\vec B_1$
(w.r.t.\ $S$),
it is sufficient that $A$ is covered by $H\gets\vec B_0$ (w.r.t.\ $S$).
The former is implied by the latter and condition \ref{def:c-covered:cond:cut}.

\end{remark}

The core of the proposed method of proving completeness is the following
sufficient condition.

\begin{theorem}
[completeness]
\label{th:complete}
Consider an LD-tree $T$ for a program $P$, and the tree $pruned(T)$.
Let $Q$ be the root of $T$.  Assume that $Q$ does not contain !.
Let $S\subseteq\HB$ be a specification, and $\pre,\post$
a c-s-specification such that $S\subseteq \post$.
Let
\begin{itemize}
\item
  $pruned(T)$ be finite, $P$ be c-s-correct w.r.t.\ $pre,post$,
  $Q$ be a well asserted query w.r.t.\ $pre,post$, and
\item
each $A\in S$ be c-covered w.r.t.\ $S, pre, post$ by a clause of $P$.
\end{itemize}
Then $pruned(T)$ is complete w.r.t.\ $S$.  
\end{theorem}

The proof is presented in the Appendix.
The next section contains example completeness proofs employing this theorem.

\paragraph{Additional comments.}
These remarks may be skipped at the first reading.

To deal with an initial query $Q$ containing the cut,
one may add a clause $p(\vec V)\gets Q$ to the program (where $\vec V$ are
the variables of $Q$, and $p$ is a new symbol),
and extend the specifications appropriately.
Specification $S$ should be extended to 
$ S' = S\cup
\{\,
 p(\vec V)\theta\in\HB \mid  Q\theta \mbox{ is ground, }
                             Q\theta\subseteq S\cup\{!\}
\,\}
$,
and all the $p$-atoms should be added to \pre and to \post.
Then Theorem \ref{th:complete} is applicable to the extended program.
(Note that
clause $p(\vec V)\gets Q$ covers each $p$-atom of $S'$ w.r.t. $S'$, 
by the definition
of $S'$, and that condition \ref{def:c-covered:cond:previous} of 
Def.\,\ref{def:c-covered} vacuously holds for  $p(\vec V)\gets Q$.
Note also that  $p(\vec V)\gets Q$ satisfies the sufficient condition for
c-s-correctness, i.e.\ is well-asserted w.r.t.\ the new c-s-specification,
as $Q$ is well-asserted w.r.t.\ $\pre,\post$.)
We skip further details.

The theorem is inapplicable to infinite pruned trees.  This restriction is
not easy to overcome: the proof of the theorem is based on constructing a
non-failing branch of the tree, the branch -- if finite -- provides the
required answer for $Q$.

\pagebreak[3]
We would like to note a technical detail:
 Def.\,\ref{def:c-covered} actually refers to
$\post\cap\HB$, not to the whole \post.

A version of  Def.\,\ref{def:c-covered} is possible,
which instead of \post employs 
a specification $S^+\in\HB$ w.r.t.\ which the program is correct;
\post is replaced by $S^+$
in conditions \ref{def:c-covered:cond:previous}, \ref{def:c-covered:cond:cut}.
(So in this version, clauses are c-covered w.r.t.\ $S$, \pre and $S^+$.)
We state without proof that
Theorem \ref{th:complete} (with obvious modifications) also holds
for such modified notion of c-covered.
(The modifications are: requiring that $P$ is correct w.r.t.\ $S^+$,
and c-s-correct w.r.t.\ $\pre,\post'$, for some $\post'$;
all other occurrences of \post are replaced by $S^+$.)
In this way a c-s-specification is used only to describe the form of atoms
selected in the derivations, and the specification $S^+$ describes 
the obtained computed answers.
We expect that such separation may be convenient in some cases.

\section{Examples}
\label{sec:examples}
This section presents three example proofs of completeness of pruned trees.
The first one considers a case where various branches produce the same answer
and some of them are pruned.
The second is a rather artificial example, to illustrate some details of 
Def.\,\ref{def:c-covered}.  
In the third example we prove that the usual way of programming negation as 
failure in Prolog correctly implements negation as finite failure for ground
queries.
\begin{example}
\label{ex:in}
Consider a program  IN:
\[
\begin{array}{l}
    \begin{array}[b]{l}
      in([\,],L). \\
      in([E|T],L) \,\pgets\, m(E,L),\, !,\, in( T, L).
    \end{array}
    \qquad
    \begin{array}[b]{l}
      m(E,[E|L]). \\
      m(E,[H|L]) \,\pgets\, m(E,L).
    \end{array}
\end{array}
\]
In the program, a single answer for $m(E,L)$ is sufficient (to obtain the
required answer for a ground $in$-atom).  So the cut is used to prune 
further answers for $m(E,L)$. 
Consider specifications
\[
\begin{array}{l}
 S = S_m\cup S_{in},\ \  pre=pre_m\cup pre_{in}, \ \ post = \TB,
 \mbox{ \ where}
\\[.5ex]
  S_m = \{\, m(t_i, [\seq t] )\in\HB \mid 1\leq i\leq n \,\},
  \\
  S_{in} = \{\, in([\seq[m]u], [\seq t] )\in\HB  \mid m,n\geq0,\ 
            \{\seq[m]u\}\subseteq\{\seq t\} \,\},
\\
pre_{m}= \{\, m(u,t)\in\TB \mid t \mbox{ is a list} \,\},
\\
pre_{in}= \{\, in(u,t)\in\HB \mid u,t \mbox{ are ground lists} \,\}.
\end{array}
\]
The program is c-s-correct w.r.t.\  $\pre,\post$ 
(by Th.\,\ref{th:c-s-correctness}, we skip rather simple details).
We show that each atom 
 $A = in( u,t )\in S_{in}$, where $u=[\seq[m]u]$, $m>0$,
 is c-covered 
by the second clause $C$ of IN.  
Note first that $A$ is covered by $C$, due to its instance 
$in([u_1|[\SEQ u 2 m]], t) \pgets m(u_1,t), !, in( [\SEQ u 2 m], t)
$;
its body atoms are in $S$, as each $u_i$ is a member of $t$.
Condition \ref{def:c-covered:cond:previous} of Def.\,\ref{def:c-covered} holds,
as $H''=A$ and $H''$ is not unifiable with $in([\,],L)$.

To check condition \ref{def:c-covered:cond:cut},
take an instance $in([E|T],L)\rho\in \pre$ of the head of $C$.
The instance is ground, and the whole $C\rho$ is ground.  
So in Def.\,\ref{def:c-covered}, $\rho\eta=\rho$.
If  $A$ is an instance of (thus equal to) $in([E|T],L)\rho$
then $in(T,L)\rho = in([u_2,\ldots,u_m],t)\in S$ (as $A\in S$).
Thus $A$ is covered by $(in([E|T],L)\pgets in(T,L))\rho\eta$.
So condition \ref{def:c-covered:cond:cut} of Def.\,\ref{def:c-covered} holds.
Thus $A$ is c-covered by $C$.  It is easy to check that 
all the remaining atoms of $S$ are covered and c-covered w.r.t.\ $S$
by the remaining clauses of IN.

Note that program IN is recurrent under the level mapping
 $|m(s,t)|=|t|$, $|in(s,t)|=|s|+|t|$, 
where
 $  |\, [h|t]\, | = 1+|t| $ and 
 $  |f(\seq t)| = 0 $
 (for any ground terms $h,t,\seq t$, and any function symbol $f$ 
  distinct from $[\ | \ ]$\,).
Thus each LD-tree for IN and a query $Q\in \pre$ is finite.

By Th.\,\ref{th:complete}, for each $Q\in \pre$ the pruned LD-tree  is
complete w.r.t.\ $S$.
Notice that condition \ref{def:c-covered:cond:cut} may not hold
when non ground arguments of $in$ are allowed in $pre_{in}$,
and that for such queries the pruned LD-trees may be not complete w.r.t.\ $S$.
As an example take $H'' = in([X],[1,2])$ and $A=in([2],[1,2])$.
\end{example}

The previous example illustrates a practical case of so called ``green cut''
\cite{Sterling-Shapiro-short},
where (for certain queries) pruning does not remove any answers.
However it represents a rather simple application of  Th.\,\ref{th:complete},
with an easy check for condition
 \ref{def:c-covered:cond:previous} of Def.\,\ref{def:c-covered},
and condition \ref{def:c-covered:cond:cut} applied only to
ground atoms from \pre.  
The next two examples illustrate more sophisticated cases of
conditions \ref{def:c-covered:cond:previous},\,\ref{def:c-covered:cond:cut}.

\begin{example}
\label{ex:artificial}
Consider a program $P$:
\[
\begin{array}[t]{l}
  p(X,Y) \,\pgets\, q(X,Y),\, r(X,Y),\, !. \\
  p(X,Z) \,\pgets\, q(X,Y),\, !,\, r(Y,Z).
\end{array}
\qquad\quad
\begin{array}[t]{l}
  q(a,a) \\ q(a,a') \\ q(b,b) 
\end{array}
\qquad\quad
\begin{array}[t]{l}
  r(a,c) \\   r(a',c) 
\end{array}
\]
\vspace{-1.\belowdisplayskip}
and specifications
\[
\begin{array}{l}
    S=\{\,p(a,c),  q(a,a'),   r(a,c),   r(a',c) \,\}, \\
    post = S\cup\{q(a,a)\}, \\
    pre = \{\, p(a,t) \mid t\in\TU \,\} \cup
     \{\, q(a,t) \mid t\in\TU \,\} \cup
     \{\, r(t,u) \mid t,u\in\TU \,\}.
\end{array}
\]
The program is c-s-correct w.r.t.\ $pre,post$, by Th.\,\ref{th:c-s-correctness}.%
\footnote{%
  Note that $S$ does not require $q(a,a)$ or $q(b,b)$ to be computed,
   and that $P$ is not correct w.r.t.\ $S$, cf.\ 
  ``Approximate specifications'' in Section \ref{sec:compl.decl}.
  Note also a usual situation:
  even if we are interested in completeness w.r.t.\ $\{p(a,c)\}$,
   some $q$- and $r$-atoms are to be present in $S$ in order to facilitate
  the proof.
}
We show that
atom  $A = p(a,c)\in S$ is c-covered by the second clause of $P$.
Note first that $A$ is covered by the clause w.r.t.\ $S$ due to its instance
$p(a,c)\pgets q(a,a'),{!}, r(a',c)$.%
\footnote{%
    Note also that each atom of $S$ is covered by $P$ w.r.t.\ $S$, and that $P$
    is recurrent.  Thus $P$ is complete w.r.t.\ $S$
    (by Th.\,\ref{th:semi-complete} and the remark following
    Th.\,\ref{th:termination}).
}

For condition \ref{def:c-covered:cond:previous} of Def.\,\ref{def:c-covered}
it is sufficient to consider $H''=p(a,X)$,
by Remark \ref{remark:conditions}.
No ground instance $p(a,s)$ of $H''$ ($s\in\HU$)
is covered by $p(X,Y) \pgets q(X,Y), r(X,Y)$ w.r.t.\ $\post\cap\HB$,
as in no ground instance of $q(X,Y), r(X,Y)$ both atoms are in $\post$.
So condition \ref{def:c-covered:cond:previous} holds.

By Remark \ref{remark:conditions},
it is sufficient to check 
condition \ref{def:c-covered:cond:cut} of Def.\,\ref{def:c-covered}
for $\rho=\{X/a\}$,
as $p(X,Z)\rho=p(a,Z)$ is a most general $p$-atom in $pre$.
If  $q(X,Y)\rho\eta \in post$ (and $do m(\eta) = \vars(q(X,Y)\rho) = \{Y\}$)
then $\eta=\{Y/a\}$ or $\eta=\{Y/a'\}$.
Hence $r(Y,Z)\rho\eta$ is $r(a,Z)$ or $r(a',Z)$.
In both cases, 
$p(a,c)\pgets r(Y\eta,c)$ is a ground instance of
$(p(X,Z)\pgets r(Y,Z))\rho\eta$ 
(i.e. of $p(a,Z)\pgets r(Y\eta,Z)$) 
covering $p(a,c)$ w.r.t.\ $S$.
So condition \ref{def:c-covered:cond:cut} holds, and $p(a,c)$ is covered by
$P$. 

The remaining atoms of $S$ are trivially c-covered by the unary clauses of $P$.
For any $A\in\pre$, 
the LD-tree for $P$ and $A$ is finite, hence the pruned LD-tree
is complete w.r.t.\ $S$ by Th.\,\ref{th:complete}.
Note that a nontrivial \post was necessary here.
With \post being \HB both conditions  \ref{def:c-covered:cond:previous},
 \ref{def:c-covered:cond:cut} do not hold.

\end{example}

\begin{example}
\label{ex:naf}
\newcommand{\mynotp}{\ensuremath{{\it not p}}\xspace}
\newcommand{\myextrahspace}{\hspace{0pt plus 5pt}}
In this example the cut is used to implement negation as finite failure.
Consider a program $P_0$ without the cut.
Assume that $P_0$ is c-s-correct w.r.t.\ a specification
$\pre_0,\post_0$, and that predicate symbols \mynotp, ${\it fail}$ do not
occur in $P_0,\pre_0,\post_0$.
Let $p$ be a unary predicate symbol.
Let $P$ be $P_0$ with the following clauses added:
\[
\mynotp(X) \,\pgets\, p(X),\, {!},\, {\it fail}.  \qquad \qquad
\mynotp(X).
\]
Let $\pre_\mynotp = \{\,\mynotp(t)\in\HB \mid  p(t)\in\pre_0 \,\}$.
Program $P$ is c-s-correct w.r.t.\ the c-s-specification $\pre,\post$, where
$\pre = \pre_0\cup \pre_\mynotp\cup\{{\it fail}\}$,
and $\post = \post_0\cup \pre_\mynotp$.
We show that $\mynotp(t)$ succeeds for those ground $t$ for which $p(t)$ is
known to finitely fail
(i.e.\ the LD-tree for $p(t)$ and $P_0$ is finite, and
$p(t)\not\in\post_0$, hence  $p(t)$ does not succeed).
Formally, the property to be proven is that the finite pruned LD-trees for
$P$ are complete w.r.t.\ a specification
$S = \{\,\mynotp(t) \in\pre_\mynotp \mid  p(t)\not\in\post_0 \,\}\cup\M_{P_0}$.

Take an atom $A = \mynotp(t) \in S$.  To show that $A$ is c-covered
w.r.t.\ $S, \pre, \post$  by
clause $\mynotp(X)$ of $P$,
condition \ref{def:c-covered:cond:previous} of Def.\,\ref{def:c-covered} and
the clause
\myextrahspace%
 $C'=\mynotp(X) \pgets p(X),\, {!},\, {\it fail}$
\myextrahspace%
 has to be considered.
Using the notation of Def.\,\ref{def:c-covered},
\,\mbox{$H''= A$}, \,$H''$ is ground, and $H''$ is not covered by 
$\mynotp(X) \pgets p(X)$ w.r.t.\ $\post\cap\HB$ (as $p(t)\not\in\post$).
Thus condition \ref{def:c-covered:cond:previous} holds and $A$ is c-covered
by a clause of $P$.  
The remaining atoms of $S$ (those from $\M_{P_0}$) are obviously covered by 
the clauses of $P_0$, as no cut occurs in $P_0$, and each atom of 
$\M_{P_0}$ is covered by $P_0$ w.r.t.\ $\M_{P_0}$.
By Th.\,\ref{th:complete} each finite pruned LD-tree for $P$ 
with an atomic root from \pre is complete w.r.t.\ $S$. 

Note that condition \ref{def:c-covered:cond:previous} may not hold when
\pre contains non-ground $\mynotp$-atoms,
and that for such query the pruned LD-tree may be not complete w.r.t.\ $S$.
(Assume that \pre contains an atom $B=\mynotp(u)$ such that (i)~$p(u)$ succeeds,
 but 
(ii)~$p(u\sigma)\not\in\post$  for some ground instance $u\sigma\in\HU$ of $u$.
 Hence $B$ fails, by~(i), and  $B\sigma\in S$, by~(ii).
 So the pruned LD-tree for $B$ is not complete w.r.t.\ $S$.  On the other hand, 
 $B\sigma$ is not c-covered by $P$ w.r.t.\ $S, \pre, \post$, as
 condition \ref{def:c-covered:cond:previous} is violated:
   Take $H''=B$,
   and a ground instance $p(u\theta)\in\HB$ of the answer for $p(u)$.
   Thus $p(u\theta)\in\post\cap\HB$.   
   So a ground instance $\mynotp(u\theta)$ of $H''$ is covered by
   $\mynotp(X) \pgets p(X)$ w.r.t.\ $\post\cap\HB$, which is forbidden 
   by condition \ref{def:c-covered:cond:previous}.)

\end{example}

\section{Conclusion}

This paper introduces a sufficient condition for completeness of Prolog
programs with the cut.  The syntax is formalized as definite clause programs
with the cut.  The operational semantics is formalized in two steps:
LD-resolution, and pruning LD-trees.
The sufficient condition is illustrated by example completeness proofs.

\paragraph{Acknowledgement.}
Thanks are due to Paul Tarau for noticing an incorrect claim in the 
fist version of this work.
The author thanks anonymous reviewers for the pinpointed errors and for
suggestions concerning the presentation.

%

%

%

%


%
\appendix
\section{Appendix}
The appendix contains a proof of Theorem \ref {th:complete}.
It also introduces a proposition
 (\ref{lemma:instances:c-covered:tough:condition})
employed in Remark \ref{remark:conditions} (to simplify checking 
condition  \ref{def:c-covered:cond:cut} of Def.\,\ref{def:c-covered}).
The proof employs the notions of unrestricted derivation and lift from 
\cite[Definitions 5.9, 5.35]{Doets}.
We first present their definitions and the main related technical result,
adjusted to the Prolog selection rule
and to the difference (explained in Section \ref{sec:LD})
between the standard notion of 
SLD-derivation adopted here and the notion of derivation of \cite{Doets}.

\begin{df}
An {\bf unrestricted LD-derivation} for a program $P$
is a (finite or infinite) sequence $Q_0,Q_1,\ldots$ of
queries, together with a sequence $\theta_1,\theta_2,\ldots$ of substitutions
(called {\em specializations}) and a sequence $C_1,C_2,\ldots$ of clauses
from $P$, such that for each $Q_i$ ($i\neq0$):
{\sloppy\par}
\begin{itemize}
\item 
when $Q_{i-1}= {!}, \vec A$ then $Q_{i}= \vec A$ and $\theta_i=\epsilon$,

\item
otherwise, when
$Q_{i-1}= A, \vec A$ then $Q_{i}= \vec B,\vec A\theta_i$, where
$do m(\theta_i)\subseteq \vars(Q_{i-1})$ and 
$A\theta_i\gets\vec B$ is an
 instance of $C_i$.

\end{itemize}
An unrestricted LD-derivation is {\bf successful} if its last query is empty.
The {\bf answer} of such successful derivation 
$Q_0,\ldots,Q_n;\theta_1,\ldots,\theta_n;C_1,\ldots,C_n$
is $Q_0\theta_1\cdots\theta_n$.
\end{df}

So, informally, the difference between derivations and unrestricted
derivations is that the latter employ clause instances which may be not most
general ones.
A technical difference is that a most general unifier in a derivation applies
both to the variables of a query and those of a clause (variant), 
while in a restricted derivation a specialization applies only to the
variables of a query.

When presenting (unrestricted) derivations,
we sometimes skip the sequence of clauses, or the sequence of substitutions.
Note that if 
$Q_0,Q_1,\ldots;\theta_1,\theta_2,\ldots$ is an LD-derivation then
$Q_0,Q_1,\ldots;\restrict{\theta_1}{Q_0},\restrict{\theta_2}{Q_1},\ldots$ 
(together with a suitable sequence of clause instances) 
is an unrestricted LD-derivation.
We say that the latter is the unrestricted LD-derivation {\em corresponding}
to the former. 
Assume that both derivations are successful.  Then both have the same answer:
$Q_0\theta_1\cdots\theta_n = 
Q_0(\restrict{\theta_1}{Q_0})\cdots(\restrict{\theta_n}{Q_{n-1}})$.

\pagebreak[3]
\begin{df}[lift]
An LD-derivation $D = Q_0,Q_1,\ldots;\theta_1,\theta_2,\ldots$
is a {\bf lift} of an unrestricted LD-derivation
$E = R_0,R_1,\ldots;\alpha_1,\alpha_2,\ldots; C_1,C_2,\ldots$ when 
\begin{itemize}
\item
$R_0$ is an instance of $Q_0$,
\item 
$D,E$ are of the same length,
\item
to each $Q_{i-1}$ a variant of a clause $C_i$ has been applied in $D$.
\end{itemize}
\end{df}

Our proof refers to the lifting theorem in the form of \cite[Th.\,5.37]{Doets}.
The full power of the theorem is not needed here, the following corollary is sufficient.

\begin{corollary}[lifting]
\label{cor:lifting}

Every unrestricted LD-derivation $E$ starting from a query $R_0$,
which is an instance of $Q_0$, 
has a lift $D$ starting from $Q_0$.

If an LD-derivation $D = Q_0,Q_1,\ldots$ is a lift of 
an unrestricted LD-derivation $E = R_0,R_1,\ldots$ then
each $R_i$ is an instance of $Q_i$.
Hence $D$ is successful iff $E$ is successful.

If an LD-derivation $D = Q_0,Q_1,\ldots;\theta_1,\theta_2,\ldots$
is a lift of an unrestricted LD-derivation
$E = R_0,R_1,\ldots;\linebreak[3]\alpha_1,\alpha_2,\ldots; C_1,C_2,\ldots$ then
$R_i\alpha_{i+1}\cdots\alpha_j$ is an instance of 
$Q_i\theta_{i+1}\cdots\theta_j$
(for any $i<j$ such that $\SEQ R i j$ are queries of $D$).
In particular, if $E$ is successful then the answer of $E$ is an instance of
the answer of $D$.

\end{corollary}

We are ready to begin a proof of Th.\,\ref {th:complete}.
It consists of a few lemmas.

\begin{lemma}
\label{lemma:antilifing}
Let $P$ be a program and $D = \SEQ Q 0 n; \seq\theta$ an LD-derivation for $P$.
Consider a substitution $\sigma$,
and the instances  $Q_0' = Q_0\theta_1\cdots\theta_n\sigma$
and $Q_n' = Q_n\sigma$ of $Q_0$ and $Q_n$.
Then there exists an unrestricted derivation 
$D' = \SEQ{Q'} 0 n$ for $P$ such that $D$ is a lift of $D'$.
\end{lemma}

\begin{proof}
The queries of $D'$ are $Q_i'=Q_{i}\theta_{i+1}\cdots\theta_n\sigma$, for
$i=0,\ldots,n$.
Let $Q_{i-1}= A,Q'$ and $Q_i = (\vec B,Q')\theta_i$, where 
$H\gets\vec B$ is a (variant of a) clause of $P$, and
$A \theta_i = H\theta_i$.
Then $Q_{i-1}'= (A,Q')\theta_{i}\cdots\theta_n\sigma$.  Applying an instance 
$(H\gets\vec B)\theta_{i}\cdots\theta_n\sigma$ of the clause,
we obtain  $(\vec B,Q')\theta_{i}\cdots\theta_n\sigma$ which is
$Q_{i}'$.
\end{proof}

\begin{lemma}
\label{lemma:sequence}
Let $S$ be a specification.
Let a program $P$ be c-s-correct w.r.t.\ a call-success
specification \pre,\,\post, and $Q_0$ be a well-asserted query.

Let
$D=\SEQ Q 0 n ; \seq\theta$ be an LD-derivation of $P$, where
$Q_0 = A,Q'$ and $Q_1 = (\vec B_0,!,\vec B_1,Q')\theta_1$
(so the first clause (variant) employed in $D$ is 
$C = H\gets \vec B_0,!,\vec B_1$).

Let
 $\SEQ Q 1 n$ be a successful subderivation of $D$ for $\vec B_0\theta_1$
(so $Q_n = (!,\vec B_1,Q')\theta_1\cdots\theta_n$).

Let
$S\models Q_0\sigma_0$
 for some ground instance  $Q_0\sigma_0$ of $Q_0$,
and  $A\sigma_0$ be c-covered by $C$ w.r.t. $S,\pre,\post$.

Then there exists a ground instance $Q_n\sigma'$ of $Q_n$ such that 
$S\models Q_n\sigma'$.
Moreover, $Q_0\sigma_0$ is the first and 
 $Q_n\sigma'$ is the last query of an unrestricted LD-derivation $D'$ such
that $D$ is a lift of $D'$.

\end{lemma}

\begin{proof}
We have
 $Q_1 = (\vec B_0,!,\vec B_1,Q')\theta_1$, and
 $Q_n = (!,\vec B_1,Q')\theta_1\cdots\theta_n$.
As $P$ is c-s-correct, $\vec B_0\theta_1\cdots\theta_n\subseteq\post\cup\{!\}$.
Without loss of generality we can assume that $do m(\sigma_0)= vars(Q_0)$.
Remember that $A\theta_1=H\theta_1$.

Let $\rho = \restrict{\theta_1}C = \restrict{\theta_1}H$.
Note that 
$A\sigma_0$ is an instance of $A$ and of $H$, as $A\sigma_0$ is covered by $C$.
Thus
$A\sigma_0$ is an instance of $H\theta_1 = H\rho$.
As $A\sigma_0$ is c-covered by $C$, 
from Def.\,\ref{def:c-covered} it follows that
$A\sigma_0$ is covered w.r.t.\ $S$ by $(H\gets\vec B_1)\rho\eta$, for any
$\eta$ such that $\vec B_0\rho\eta$ is ground, 
$\vec B_0\rho\eta\subseteq \post\cup\{!\}$ and $do m(\eta) = \vars(\vec B_0\rho)$.
In particular, this holds for
$\eta = \restrict{ (\theta_2\cdots\theta_n\tau) }{\vec B_0\theta_1}$
(for any $\tau$ for which
 $\vec B_0\rho\eta = \vec B_0\theta_1\cdots\theta_n\tau$ is ground,
and $do m(\tau)\subseteq\vars(\vec B_0\theta_1\cdots\theta_n)$
).

Note that $\eta = 
\left( \restrict{ (\theta_2\cdots\theta_n) }{\vec B_0\theta_1} \right) \tau
$.
Also
$\vec B_0\rho = \vec B_0\theta_1$, and
$(H\gets\vec B_1)\rho= (A\gets\vec B_1)\theta_1$.
Thus
$\mbox{$\vars((H\gets\vec B_1)\rho)$}\subseteq
 \vars(A)\cup\vars(\vec B_1)\cup\vars(\theta_1)$.
By Lemma \ref{lemma:subderivation-variable} applied to $Q_1$,
if a variable $X$ occurs in $(H\gets\vec B_1)\rho$ but not in $\vec B_0\theta_1$
then $X$ does not occur in $\SEQ \theta 2 n$.  Hence 
$X\theta_2\cdots\theta_n = X$, so
$X\eta = X\tau = X\theta_2\cdots\theta_n\tau$, and thus
$(H\gets\vec B_1)\rho\eta = (H\gets\vec B_1)\theta_1\cdots\theta_n\tau
$.
{\sloppy\par}

Let $\varphi=\theta_1\cdots\theta_n$.
As $A\sigma_0$ is covered by $(H\gets\vec B_1)\varphi\tau$, 
for some ground instance $(H\gets\vec B_1)\varphi\tau\tau'$ we have
$H\varphi\tau\tau' = A\varphi\tau\tau'= A\sigma_0$ and
 $B_1\varphi\tau\tau'\subseteq S\cup\{!\}$,
where $do m(\tau') = \vars((H\gets\vec B_1)\varphi\tau)$.

Now (i) $X\sigma_0 = X\varphi\tau\tau'$ for each variable $X$ occurring in $A$,
as $A\sigma_0 = A\varphi\tau\tau'$.
On the other hand, we show that (ii) if $X$ occurs in $Q_0$ but not in $A$ then
$X\varphi\tau\tau' = X$.
Take such variable $X$. 
By Lemma \ref{lemma:subderivation-variable} applied to $Q_0$,
variable X occurs neither in  $(\vec B_0,\vec B_1)\theta_1$, nor in any of
the mgu's \seq\theta.  Hence $X$ does not occur in
$\vec B_0\varphi$, $\vec B_1\varphi$, $A\varphi$.
Thus 
$X\varphi\tau\tau' = X\tau\tau' = X\tau' = X$
(as $X\not\in do m(\varphi)$,
 $X\not\in do m(\tau) \subseteq\vars(\vec B_0\varphi)$,
 $X\not\in do m(\tau') \subseteq\vars((A{\gets}\vec B_1)\varphi)$\,).
This completes the proof of (ii).

    Let us split $\sigma_0$, let $\sigma_1=\restrict{\sigma_0}A$ and
    $\sigma_2=\sigma_0\setminus\sigma_1$.
    As $\sigma_0$ is ground, 
     $\sigma_0 = \sigma_1\sigma_2= \sigma_2\sigma_1$.
Now by (i), for any $X\in\vars(A)$ we have 
$X\sigma_0 = X\varphi\tau\tau' = X\varphi\tau\tau'\sigma_2$ (as $X\sigma_0$ is
ground).  For any $X\in\vars(Q_0)\setminus\vars(A)$ we have 
$X\sigma_0 = X\sigma_2$;  by (ii) it follows
$X\sigma_0 = X\varphi\tau\tau'\sigma_2$.
Thus $Q_0\sigma_0 = Q_0\varphi\tau\tau'\sigma_2$.
Let $\psi=\tau\tau'\sigma_2$.
We have
 $S \models Q_n\psi$,
as  $Q_n\psi$ consists of 
 ${!},\vec B_1\varphi\psi$
and of $Q'\varphi\psi$, which is a fragment of  $Q_0\sigma_0$;
we showed that  $B_1\varphi\tau\tau'\subseteq S\cup\{!\}$ 
(hence $B_1\varphi\psi\subseteq S\cup\{!\}$),
also $S\models Q_0\sigma_0$ holds by the premises of the Lemma.

The required instance $Q_n\sigma'$ is  $Q_n\psi$.
The required unrestricted derivation $D'$ exists by Lemma 
\ref{lemma:antilifing}
(with $Q_0' = Q_0\sigma_0 = Q_0\varphi\psi$ and $Q_n' =  Q_n\psi$).
\end{proof}

\begin{lemma}
\label{lemma:main}
  Let $T, P, Q, S, \pre, \post$ be as in Th.\,\ref {th:complete}.
Consider a node $Q_0$ of $pruned(T)$ with a ground instance $Q_0\sigma$,
such that $Q_0$ is not empty, $S\models Q_0\sigma$, 
and for $Q_0$ it holds that
\begin{equation}
\label{node.condition}
  \begin{array}{l}
\mbox{if the node occurs in a cutting sequence $D$ of }  pruned(T) \\  
\mbox{then it is the introducing node of } D.
  \end{array}
\end{equation}
Then there exists in $pruned(T)$ a descendant $Q_k$ of $Q_0$
satisfying (\ref{node.condition}),
with a ground instance $Q_k\sigma'$ such that $S\models Q_k\sigma'$. 
Moreover, $Q_0\sigma$, $Q_k\sigma'$ are, respectively, the first and the last
query of an unrestricted LD-derivation $D'$ for $P$, with a lift 
$D = Q_0,\ldots,Q_k$
(where $D$ is the LD-derivation consisting of the
nodes between $Q_0$ and $Q_k$ in $pruned(T)$\,).

\end{lemma}

\begin{proof}
Outline:
We first show that $Q_0$ has a child $Q$ in $T$ such that $S\models\exists Q$.
Then we show that $Q$ is a node of $pruned(T)$.
If $Q$ does not occur in a cutting sequence of $pruned(T)$ then 
the required node $Q_k$ is $Q$.
Otherwise $Q$ is the second node of a cutting sequence $D_0$ beginning in $Q_0$.
In this case the required node $Q_k$ is the child of the last node of $D$.
The details follow below.

Let $Q_0=A,Q'$.
Now $A\sigma\in S$, so $A\sigma$ is c-covered, hence covered, w.r.t.\ $S$
by a clause $C = H\gets\vec B$ of $P$.
Node $Q_0$ has a child $Q= (\vec B,Q')\theta$ in $T$
(where $\theta$ is an mgu of $A$ and $H$).
Let $A\sigma\gets\vec B\sigma'$ be a ground instance of $C$ such that 
the atoms of $\vec B\sigma'$ are in $S\cup\{!\}$.
Let $Q''$ be $\vec B\sigma',Q'\sigma$.
Obviously, $S\models Q''$,
and $Q_0\sigma, Q''$ is an unrestricted LD-derivation for $P$.
Now LD-derivation $Q_0,Q$ is its lift and, by Lifting Corollary
\ref{cor:lifting}, $Q''$ is an instance of $Q$.

Assume that $Q_0$ does not occur in any cutting sequence of $pruned(T)$.
Then $Q$ is a node of $pruned(T)$, and
$Q$ is the required descendant of $Q_0$.  Moreover,
derivation $Q_0,Q$ is a lift of $Q_0\sigma,Q''$.

It remains to consider the case of  $Q_0=A,Q'$ being the introducing node of a
cutting sequence $D_0=\SEQ Q 0 j$ of $pruned(T)$. 
Note that if in  $pruned(T)$  there are two such sequences then 
one of them is a prefix of the other.%
\footnote{%
As if a node $R$ occurs in a cutting sequence of $pruned(T)$ then at most one
of its children occurs in a cutting sequence of $pruned(T)$.
(Assume that two children do; then one of them is pruned due to the other one.)
}
(Such distinct sequences are possible when there are multiple cuts in a
clause applied to $Q_0$.)
Let us assume that $D_0$ is the longest cutting sequence with the introducing
node $Q_0$.  Let $\seq[j]\theta$ be the sequence of mgu's of $D_0$ viewed as
an LD-derivation.

So $Q_1= (\vec B_0,!,\vec B_1,Q')\theta_1$,
where $H\gets \vec B_0,!,\vec B_1$ is a (variant of a) clause of $P$,
and $Q_j=(!,\vec B_1,Q')\theta_1\cdots\theta_j$.
From c-s-correctness of $P$ w.r.t.\ $\pre,\post$ it follows that
each selected atom in $D_0$ is in \pre, and the answer in $D_0$ for this atom 
is in \post.  As post is closed under substitution, 
each atom of $\vec B_0\theta_1\cdots\theta_j$ is in $\post\cup\{!\}$.

We are ready to show
that the node $Q$ is present in $pruned(T)$.  Assume it is not.
So $Q$ occurs to the right of $D_0$ in $T$, 
as $Q$ is a child of $Q_0$, and $Q_0$ occurs in $pruned(T)$.
Atom $A\sigma$ is c-covered by $C$, which is preceded in $P$ by 
(a variant of) clause $H\gets \vec B_0,!,\vec B_1$.
Let us make explicit the first occurrence of the cut in the clause:
let $ \vec B_0,!,\vec B_1 =  \vec A_0,!,\vec A_1$,
where $\vec A_0$ does not contain a ! and
is a prefix of $\vec B_0$.

Consider a ground instance 
$C'=(A\gets\vec A_0)\theta_1\cdots\theta_j\rho$ of $H\gets\vec A_0$
(the former is an instance of the latter as $A\theta_1 = H\theta_1$).
Now $A\sigma$ is an instance of $A\in\pre$, and a ground instance 
$A\theta_1\cdots\theta_j\rho$ of $A$ is covered by $H\gets\vec A_0$ w.r.t.\ 
$\post\cap\HB$
(as each atom of the body of clause $C'$ is in $(\post\cap\HB)\cup\{!\}$).
Thus condition \ref{def:c-covered:cond:previous} of Def.\,\ref{def:c-covered}
is violated, and $A\sigma$ is not c-covered by $C$; contradiction.
Hence either $Q$ occurs in $pruned(T)$ to the left of $Q_1$, or $Q=Q_1$.
In the first case (i.e. $Q\neq Q_1$), the node $Q$ is the required descendant
of $Q_0$ (with $Q_0,Q$ being a lift of $Q_0\sigma,Q''$, as shown above).

In the second case (where $Q=Q_1$), by Lemma \ref{lemma:sequence},
node $Q_j$ has an instance $Q_j\sigma'$ such that $S\models Q_j\sigma'$,
there exists an unrestricted LD-derivation $D_0'=Q_0\sigma,\ldots,Q_j\sigma'$,
and $D_0$ is a lift of $D_0'$.
Now the child $Q_k= (\vec B_1,Q')\theta_1\cdots\theta_j$ of $Q_j$ 
is the required node of $pruned(T)$,
$D'$ is $D_0',Q_k\sigma'$, and its lift $D$ is $D_0,Q_k$.
\end{proof}

\smallskip
\begin{proof}
[Proof of Th.\,\ref {th:complete}.]

Let $Q_0$ be the root of $T$, 
and $Q_0\sigma$ be its ground instance such that $S\models Q\sigma$.
As no ! occur in $Q_0$, $Q_0$ satisfies (\ref{node.condition}).
By induction from Lemma \ref{lemma:main} we obtain that there is a successful
or infinite branch $D$ in $pruned(T)$, which is a lift of an unrestricted
derivation $D'$ beginning with $Q_0\sigma$.  As $pruned(T)$ is finite, $D, D'$ are
successful.
Hence the answer $Q_0\sigma$ of $D'$ is an instance of the answer of $D$.
\end{proof}

We conclude with a proposition which simplifies
checking that an atom is c-covered by a clause containing the cut
(condition \ref{def:c-covered:cond:cut} of Def.\,\ref{def:c-covered}).

\begin{propo}
\label{lemma:instances:c-covered:tough:condition}
Assume the notation of Def.\,\ref{def:c-covered}.
If condition \ref{def:c-covered:cond:cut} of Def.\,\ref{def:c-covered} 
holds for an atom $H\rho\in pre$ then it holds for
any instance $H\rho'$ of $H\rho$
such that $A$ is an instance of $H\rho'$, and $\rho'$ satisfies the
requirements of condition   \ref{def:c-covered:cond:cut}
(i.e.\
 $do m(\rho')\subseteq vars(H)$, $r n g(\rho') \cap vars(C) \subseteq vars(H)$, 
 $do m(\rho')\cap r n g(\rho') = \emptyset$).
\end{propo}

\renewcommand\AA{{\ensuremath{\vec A}}\xspace}

\begin{lemma}
\label{lemma:instances-derivations}
Let 
$C$ be a clause $H\gets\vec B_0,!,\vec B_1$.
Let $S$, $A$, $\rho,\ \eta$, $H\rho $ and $\vec B_0\rho\eta$ be as in condition 
\ref{def:c-covered:cond:cut} of Def.\,\ref{def:c-covered}.
Let $\vec B_0 =   A_1,\ldots,A_{k-1}$ and 
$\vec B_1 =   A_k,\ldots,A_{n}$
The following conditions (1) and (2) are equivalent.

\smallskip\noindent
(1) 
$A$ is covered by $(H\gets \vec B_1)\rho\eta$ w.r.t.\ $S$.

\smallskip\noindent
(2) There exists a successful LD-derivation for $A$ using in its consecutive
steps the clauses $C$,  $ A_1\rho\eta,\ldots,A_{k-1}\rho\eta$, then !, and then
some atoms from $S\cup\{!\}$.
\end{lemma}
Note that in (2) all the queries and all the clauses used in the derivation,
except $C$, are ground. 
\pagebreak[3]

\begin{proof}
\mbox{(1) $\Rightarrow$ (2)}:
\hspace{0pt plus .5em}%
(1) implies that 
 $A$ is covered by a ground clause $(H\gets \vec B_1)\rho\eta\sigma$.
Construct
an LD-derivation $D$ for $A$, using first clause $C\rho\eta\sigma$ and then the
clauses as in (2).  Its lift is a required derivation.
{\sloppy\par}

\medskip\noindent
(2) $\Rightarrow$ (1):
Take a derivation as in (2):
\[
\begin{array}{l@{\qquad}l}
  A	\\
  (\seq A)\theta_1                    & \theta_1	\\
  \ldots		              & \ldots          \\
  (!,\SEQ A {k} n)\theta_1\cdots\theta_k  & \theta_k	\\  
  \ldots		              & \ldots          \\
  A_n\theta_1\cdots\theta_{n+1}        & \theta_{n+1}	\\  
  \raisebox{-.5ex}{$\Box$}                                & \theta_{n+2}
\end{array}
\]
Its mgu's are ground substitutions.
We have
$A=H\theta_1 = H \SEQC \theta 1 {n+2}$, 
and the ground clauses used in the derivation are
$A_i\SEQC\theta 1 {i+1} = A_i\SEQC \theta 1 {n+2}$  ($i=1,\ldots,k-1$),
and
$A_i\SEQC\theta 1 {i+2} = A_i\SEQC \theta 1 {n+2}$  ($i=k,\ldots,n$). 
Comparing this with (2) gives
$A_i\SEQC \theta 1 {n+2} = A_i\rho\eta$,  for $i=1,\ldots,k-1$,
and   $A_i\SEQC \theta 1 {n+2}\in S$ for  $i=k,\ldots,n$.
Hence  $\vec B_0 \rho\eta = \vec B_0 \SEQC \theta 1 {n+2}$.
The rest of the proof, roughly speaking, deals with representing 
substitution $\SEQC \theta 1 {n+2}$ as a composition of $\rho\eta$
and a certain substitution $\sigma$.  As a result, we obtain a ground
instance of $(H\gets\vec B_1)\rho\eta$ which covers $A$.

Now $A = H\rho\delta$  for some ground substitution $\delta$ with 
$do m(\delta)=vars(H\rho)$. 
So 
$\theta_1=\restrict{(\rho\delta)}{vars(H)}$, as $do m(\theta_1)=vars(H)$.
Note that $do m(\delta) \cap vars(C)\subseteq vars(H)$
(as $do m(\delta)\subseteq\vars(H)\cup r n g(\rho)$, and
from Def.\,\ref{def:c-covered} we have
\linebreak[3]
{$r n g(\rho) \cap \vars(C)\subseteq \vars(H)$}).
Hence
$\theta_1=\restrict{(\rho\delta)}{vars(H)} = 
\restrict{(\rho\delta)}{{\it vars}(C)}$, and thus
$C\theta_1 = C\rho\delta$.
In particular, $\vec B_0\theta_1 = \vec B_0\rho\delta$.
So  $\vec B_0\rho\eta 
 = \vec B_0\SEQC\theta1{n+2}
 = \vec B_0\rho\delta \SEQC\theta2{n+2} 
$.
Thus
$\eta
      = \restrict{(\delta \SEQC \theta 2 {n+2})}{\vec B_0\rho}
$
(as $do m(\eta)=vars(\vec B_0\rho)$).

Let
$\sigma = {(\delta \SEQC \theta 2{n+2})} \setminus \eta$.
As $\eta$ and $\sigma$ are ground and with disjoint domains,
$\delta \SEQC \theta 2{n+2} = \eta\cup\sigma = \eta\sigma$.
Hence
$C\SEQC \theta 1{n+2}= C\rho\delta\SEQC \theta 2{n+2} = C\rho\eta\sigma$
(as $C\theta_1 = C\rho\delta$).
So
$H\rho\eta\sigma= H\SEQC \theta 1{n+2} = A$ and 
$A_i\rho\eta\sigma = A_i\SEQC \theta 1{n+2}\in S$, for $i=k,\ldots,n$.
Hence
$A$ is covered by $(H \gets A_k,\ldots,A_n)\rho\eta$ w.r.t.\ $S$.
\end{proof}

\begin{proof}[Proof of Proposition
   \ref{lemma:instances:c-covered:tough:condition}.]
Let $C$, the clause used in condition \ref{def:c-covered:cond:cut}, be
$H\gets\vec B_0,!,\vec B_1$.  Let  $\vec B_0$ be $\seq[{k-1}]A$.  
We first show that $\vec B_0\rho'$ is an instance of $\vec B_0\rho$.
For some~$\delta$ with $do m(\delta)\subseteq vars(H\rho)$, we have
$H\rho'=H\rho\delta$, so $\rho'=\restrict{(\rho\delta)}{H}$.
Consider a variable $X$ from $C$. There are two cases:
\\
1.  $X\in vars(H)$, thus  $X\rho'=X\rho\delta$. \
\\
2.  $X\not\in vars(H)$. So $X\rho'=X$.
Also $X\not\in do m(\rho)$, as $do m(\rho) \subseteq vars(H)$.
From $do m(\delta)\subseteq vars(H)\cup r n g(\rho)$ it follows that
 $do m(\delta)\cap vars(C)\subseteq vars(H)$
(as $r n g(\rho) \cap vars(C) \subseteq vars(H)$).
So $X\not\in do m(\delta)$.  Hence $X\rho\delta=X$ and $X\rho'=X=X\rho\delta$.

   We showed that $\rho' = \restrict{(\rho\delta)}C$.
So $\vec B_0\rho'=\vec B_0\rho\delta $.
Then each ground instance $\vec B_0\rho'\eta'$ of $\vec B_0\rho'$  
such that $\vec B_0\rho'\eta'\in\post$
is an instance of $\vec B_0\rho$
($\vec B_0\rho'\eta'= \vec B_0\rho\delta\eta'= \vec B_0\rho\eta$ where 
 $\eta=(\restrict{\delta}{ \mbox{\scriptsize$\vec B_0$}\rho})\eta'$).

Assume that condition  \ref{def:c-covered:cond:cut} holds for $H\rho$.
Then for each ground instance  $\vec B_0\rho'\eta'$ as above, where each atom of
$\vec B_0\rho\eta$ is in $post\cup\{!\}$,
atom $A$ is covered w.r.t.\ $S$ by $(H\gets \vec B_1)\rho\eta$.
By Lemma \ref{lemma:instances-derivations} 
there exists a successful LD-derivation for $A$ using in its consecutive
steps the clauses $C$,  $ A_1\rho\eta,\ldots,A_{k-1}\rho\eta$, and then
some atoms from $S\cup\{!\}$.
As $ A_i\rho\eta= A_i\rho'\eta'$ for $i=1,\ldots,k-1$,
by Lemma \ref{lemma:instances-derivations} used in the opposite direction,
$A$ is covered by  $(H\gets \vec B_1)\rho'\eta'$.
\end{proof}

{\enlargethispage*{1.2\baselineskip}}
\bibliographystyle{alpha}
\bibliography{bibshorter,bibpearl,bibmagic,bibcut}

\end{document}

%% file: diagram.v2.tex
%
    $
  \begin{array}[t]{c}
    p    \\
    \arcaa  \qquad  \arcbb \\
    \makebox[1em][c]{%
        \begin{array}[t]{c}
          q,{!}                \\
        \arcaa  \quad  \arcbb \\
            \begin{array}[t]{c}
              s,{!},r,{!}    \\
              \arcv            \\
              {!},r,{!}    \\
              \arcv            \\
              r,{!}        \\
            \end{array}
          \qquad\quad
            \begin{array}[t]{c}
              t,{!}    \\
              \arcv            \\
             {!}        \\
              \arcv       \\
              \Box
            \end{array}
            \makebox[.5em][l]{%
                \begin{array}[t]{@{}l}
                    \\
                    \arcbb    \\
                    \quad\ r,r
                \end{array}
                } 
        \end{array}
    }    
\qquad\qquad\quad
r
  \end{array}
$

%% file: complcut2.bezk.bezn.bbl
\newcommand{\etalchar}[1]{$^{#1}$}
\begin{thebibliography}{SGS{\etalchar{+}}10}

\bibitem[And03]{DBLP:journals/tplp/Andrews03}
James~H. Andrews.
\newblock The witness properties and the semantics of the {Prolog} cut.
\newblock {\em {TPLP}}, 3(1):1--59, 2003.

\bibitem[AP93]{AP93}
K.~R. Apt and D.~Pedreschi.
\newblock Reasoning about termination of pure {P}rolog programs.
\newblock {\em Information and Computation}, 106(1):109--157, 1993.

\bibitem[Apt97]{Apt-Prolog}
K.~R. Apt.
\newblock {\em From Logic Programming to {P}rolog}.
\newblock International Series in Computer Science. Prentice-Hall, 1997.

\bibitem[BC89]{DBLP:conf/tapsoft/BossiC89}
A.~Bossi and N.~Cocco.
\newblock Verifying correctness of logic programs.
\newblock In J.~D\'{\i}az and F.~Orejas, editors, {\em TAPSOFT, Vol.2}, volume
  352 of {\em Lecture Notes in Computer Science}, pages 96--110. Springer,
  1989.

\bibitem[Bez93]{DBLP:journals/jlp/Bezem93}
M.~Bezem.
\newblock Strong termination of logic programs.
\newblock {\em J. Log. Program.}, 15(1{\&}2):79--97, 1993.

\bibitem[Bil90]{DBLP:journals/tcs/Billaud90}
Michel Billaud.
\newblock Simple operational and denotational semantics for prolog with cut.
\newblock {\em Theor. Comput. Sci.}, 71(2):193--208, 1990.

\bibitem[Cla79]{Clark79}
K.~L. Clark.
\newblock Predicate logic as computational formalism.
\newblock Technical Report 79/59, Imperial College, London, December 1979.

\bibitem[Der93]{DBLP:journals/tcs/Deransart93}
P.~Deransart.
\newblock Proof methods of declarative properties of definite programs.
\newblock {\em Theor. Comput. Sci.}, 118(2):99--166, 1993.

\bibitem[DM88]{DM88}
W.~Drabent and J.~Ma{\l}uszy\'{n}ski.
\newblock {Inductive assertion method for logic programs}.
\newblock {\em Theoretical Computer Science}, 59:133--155, 1988.

\bibitem[DM93]{Deransart.Maluszynski93}
P.~Deransart and J.~Ma{\l}uszy\'nski.
\newblock {\em A Grammatical View of Logic Programming}.
\newblock The MIT Press, 1993.

\bibitem[DM05]{DBLP:journals/tplp/DrabentM05shorter}
W.~Drabent and M.~Mi{\l}kowska.
\newblock Proving correctness and completeness of normal programs -- a
  declarative approach.
\newblock {\em TPLP}, 5(6):669--711, 2005.

\bibitem[Doe94]{Doets}
K.~Doets.
\newblock {\em From Logic to Logic Programming}.
\newblock The MIT Press, Cambridge, MA, 1994.

\bibitem[Dra14]{drabent.arxiv.coco14}
W.~Drabent.
\newblock Correctness and completeness of logic programs.
\newblock {\em CoRR}, 2014.
\newblock {\small{\url{http://arxiv.org/abs/1412.8739}}}. Final version to
  appear in {\em ACM Transactions on Computational Logic}.

\bibitem[Dra15a]{drabent.lopstr14}
W.~Drabent.
\newblock On completeness of logic programs.
\newblock In {\em Logic Based Program Synthesis and Transformation, LOPSTR
  2014. Revised Selected Papers}, volume 8981 of {\em Lecture Notes in Computer
  Science}. Springer, 2015.
\newblock Extended version in {\em CoRR} abs/1411.3015 (2014).
  \small{\url{http://arxiv.org/abs/1411.3015}}.

\bibitem[Dra15b]{drabent.Herbrand.arxiv.2015}
W.~Drabent.
\newblock {On definite program answers and least Herbrand models}.
\newblock {\em CoRR}, abs/1503.03324, 2015.
\newblock \small{\url{http://arxiv.org/abs/1503.03324}}. To appear in {\em
  Theory and Practice of Logic Programming}.

\bibitem[dV89]{DBLP:journals/scp/Vink89}
E.~P. de~Vink.
\newblock Comparative semantics for {PROLOG} with cut.
\newblock {\em Sci. Comput. Program.}, 13(1):237--264, 1989.

\bibitem[KK14]{DBLP:conf/flops/KrienerK14}
Jael Kriener and Andy King.
\newblock Semantics for {Prolog} with cut - revisited.
\newblock In Michael Codish and Eijiro Sumii, editors, {\em Functional and
  Logic Programming - 12th International Symposium, {FLOPS} 2014, Kanazawa,
  Japan, June 4-6, 2014. Proceedings}, volume 8475 of {\em Lecture Notes in
  Computer Science}, pages 270--284. Springer, 2014.

\bibitem[Llo87]{lloyd87}
J.~W. Lloyd.
\newblock {\em Foundations of Logic Programming}.
\newblock Springer, 1987.
\newblock Second, extended edition.

\bibitem[Mah88]{DBLP:books/mk/minker88/Maher88}
Michael~J. Maher.
\newblock Equivalences of logic programs.
\newblock In Jack Minker, editor, {\em Foundations of Deductive Databases and
  Logic Programming.}, pages 627--658. Morgan Kaufmann, 1988.

\bibitem[NM95]{nilsson.maluszynski.book}
U.~Nilsson and J.~Ma{\l}uszy\'nski.
\newblock {\em Logic, Programming and Prolog \/{\rm(2ed)}}.
\newblock Previously published by John Wiley \& Sons Ltd, 1995.
\newblock \url{http://www.ida.liu.se/~ulfni/lpp/}.

\bibitem[SGS{\etalchar{+}}10]{DBLP:journals/tplp/Schneider-KampGSST10}
Peter Schneider{-}Kamp, J{\"{u}}rgen Giesl, Thomas Str{\"{o}}der, Alexander
  Serebrenik, and Ren{\'{e}} Thiemann.
\newblock Automated termination analysis for logic programs with cut.
\newblock {\em {TPLP}}, 10(4-6):365--381, 2010.

\bibitem[Sha83]{Shapiro.book}
E.~Shapiro.
\newblock {\em Algorithmic Program Debugging}.
\newblock The MIT Press, 1983.

\bibitem[Spo00]{DBLP:journals/jlp/Spoto00}
Fausto Spoto.
\newblock Operational and goal-independent denotational semantics for {Prolog}
  with cut.
\newblock {\em J. Log. Program.}, 42(1):1--46, 2000.

\bibitem[SS94]{Sterling-Shapiro-short}
L.~Sterling and E.~Shapiro.
\newblock {\em The Art of Prolog}.
\newblock The MIT Press, 2 edition, 1994.

\end{thebibliography}
